\documentclass[11pt]{article}
\pdfoutput=1

\usepackage[a4paper,margin=1in]{geometry}
\usepackage{xspace}
\usepackage[bookmarks,bookmarksopen,bookmarksdepth=2]{hyperref}
\usepackage{comment}
\usepackage{amsmath}
\usepackage{amsfonts}
\usepackage{amssymb}
\usepackage{amsthm}
\usepackage{thmtools,thm-restate}
\usepackage[textsize=small]{todonotes}
\usepackage{tikz}
\usetikzlibrary{arrows,automata}
\usepackage{stmaryrd}
\usepackage{mathdots}
\usepackage{algorithm}
\usepackage[noend]{algpseudocode}
\usepackage{thmtools}
\usepackage{lineno}
\usepackage{enumerate}
\usepackage{mathtools} 
\usepackage{wrapfig}

\declaretheorem{proposition}
\declaretheorem[sibling=proposition]{lemma}
\declaretheorem[sibling=proposition]{theorem}
\declaretheorem[sibling=proposition]{corollary}
\declaretheorem[sibling=proposition]{conjecture}
\declaretheorem[style=remark,sibling=proposition]{example}
\declaretheorem[style=remark,sibling=proposition]{claim}
\declaretheorem[style=remark,sibling=proposition]{remark}

\begin{document}

\title{Reachability in One-Dimensional Pushdown Vector Addition Systems is Decidable}

\author{
Clotilde Bizière \\ Université de Bordeaux, LaBRI \\ \texttt{clotilde.biziere@u-bordeaux.fr} \and
Wojciech Czerwiński\thanks{Supported by the ERC grant INFSYS, agreement no. 950398.} \\ University of Warsaw \\ \texttt{wczerwin@mimuw.edu.pl}
}



\date{}

\pagestyle{plain}


\newcommand{\clotilde}[1]{}
\newcommand{\wojtek}[1]{}

\newcommand{\ignore}[1]{}

%
%

\newcommand{\Oo}{\mathcal{O}}
\newcommand{\D}{\mathbb{D}}
\newcommand{\F}{\mathcal{F}}
\newcommand{\W}{\mathcal{W}}
\newcommand{\A}{\mathcal{A}}
\newcommand{\B}{\mathcal{B}}
\newcommand{\I}{\mathcal{I}}
\renewcommand{\S}{\mathcal{S}}
\newcommand{\M}{\mathcal{M}}
\newcommand{\N}{\mathbb{N}}
\newcommand{\Nn}{\mathcal{N}}
\newcommand{\Pp}{\mathcal{P}}
\newcommand{\Z}{\mathbb{Z}}
\newcommand{\R}{\mathbb{R}}
\newcommand{\Q}{\mathbb{Q}}
\newcommand{\Qpos}{\mathbb{Q}_+}
\renewcommand{\P}{\mathbb{P}}

\newcommand{\set}[1]{\{#1\}}
\newcommand{\setof}[2]{\set{#1 \mid #2}}
\newcommand{\card}[1]{\left|#1\right|}
\newcommand{\prettyexists}[2]{\exists_{#1} \, #2}
\newcommand{\prettyforall}[2]{\forall_{#1} \, #2}
\newcommand{\source}[1]{\textsc{source}(#1)}
\newcommand{\target}[1]{\textsc{target}(#1)}
\newcommand{\runs}[1]{\textsc{Runs}(#1)}  
\newcommand{\runord}{\unlhd}
\newcommand{\fin}{\textup{fin}}
\newcommand{\prof}{\textsc{profile}}
\newcommand{\lab}{\textsc{label}}
\newcommand{\parimage}{\textsc{pi}}
\newcommand{\rsmall}{\textsc{small}}
\newcommand{\rbig}{\textsc{big}}
\newcommand{\reachvas}{\textsc{Reach}(\mathrm{VAS})}
\newcommand{\secreachvas}{\textsc{SecReach}(\mathrm{VAS})}
\newcommand{\reachvass}{\textsc{Reach}(\mathrm{VASS})}
\newcommand{\secreachvass}{\textsc{SecReach}(\mathrm{VASS})}
\newcommand{\reach}[1]{R_{#1}}
\newcommand{\secreach}{\textsc{SecReach}}
\newcommand{\vas}{VAS\xspace}
\newcommand{\vases}{{\vas}es\xspace}
\newcommand{\petrinet}{PN\xspace}
\newcommand{\vass}{VASS\xspace}
\newcommand{\vasses}{{\vass}es\xspace}
\newcommand{\posit}{\textsc{pos}}
\newcommand{\negat}{\textsc{neg}}
\newcommand{\pre}{\textsc{pre}}
\newcommand{\post}{\textsc{post}}
\renewcommand{\sec}[2]{\textsc{sec}_{#1}(#2)}
\newcommand{\cc}{\textsc{cc}}
\newcommand{\trans}[1]{\stackrel{#1}{\longrightarrow}}
\newcommand{\tran}{\longrightarrow}
\newcommand{\Tran}{\Longrightarrow}
\newcommand{\llra}{\longleftrightarrow}
\newcommand{\To}{\Rightarrow}

\newcommand{\up}{\textup{up}}
\newcommand{\trg}{\textup{trg}}
\newcommand{\src}{\textup{src}}
\newcommand{\depth}{\textup{depth}}
\newcommand{\down}{\textup{down}}
\newcommand{\negative}{\textsc{negative}}
\newcommand{\zero}{\textsc{zero}}
\newcommand{\pref}{\textup{pref}}
\newcommand{\suff}{\textup{suff}}
\newcommand{\midd}{\textup{mid}}
\newcommand{\low}{\textup{Low}}
\newcommand{\topp}{\textup{Top}}
\newcommand{\high}{\textup{high}}
\newcommand{\init}{\textup{init}}
\newcommand{\rank}{\textup{rank}}
\newcommand{\tree}{\textup{tree}}
\newcommand{\alp}{\textup{alph}}
\newcommand{\fresh}{\textup{fresh}}
\newcommand{\dis}{\textup{dis}}
\newcommand{\ratio}{\textup{ratio}}
\newcommand{\norm}{\textup{norm}}
\newcommand{\supp}{\textup{supp}}
\newcommand{\trash}{\textup{trash}}
\newcommand{\proj}{\textup{proj}}
\newcommand{\conf}{\textup{Conf}}
\newcommand{\atrans}{\textup{AnTrans}}
\newcommand{\inp}{\textup{input}}
\newcommand{\out}{\textup{output}}
\newcommand{\size}{\textup{size}}
\newcommand{\flush}{\textbf{\textup{flush}}}
\newcommand{\ztest}{\textbf{\textup{zero-test}}}
\newcommand{\mult}{\textbf{\textup{multiply}}}
\newcommand{\reaches}{\longrightarrow}
\newcommand{\nreaches}{\centernot\longrightarrow}
\newcommand{\lin}{\textup{Lin}}
\newcommand{\ds}{\textup{data}}
\newcommand{\alleven}{\textup{AllEven}\xspace}
\newcommand{\allodd}{\textup{AllOdd}\xspace}
\newcommand{\limsupeven}{\textup{LimsupEven}\xspace}
\newcommand{\limsupodd}{\textup{LimsupOdd}\xspace}

\newcommand{\addeq}{\mathrel{+}=}
\newcommand{\subeq}{\mathrel{-}=}
\newcommand\tab{\;\;\;\;}
\newcommand{\loopp}{\textup{\textbf{loop}}\xspace}
\newcommand{\body}{\textup{\textbf{body}}\xspace}
\newcommand{\dop}{\textup{\textbf{do}}\xspace}
\newcommand{\downto}{\textup{\textbf{downto}}\xspace}
\newcommand{\ampl}{\textup{\textbf{ampl}}\xspace}

\newcommand{\eff}{\textup{eff}}
\newcommand{\leff}{\textup{left-eff}}
\newcommand{\reff}{\textup{right-eff}}
\newcommand{\maxeff}{\textup{max-eff}}
\newcommand{\maxreach}{\textup{max-reach}}
\newcommand{\pump}{\textup{pump}}
\newcommand{\lleft}{\textup{left}}
\newcommand{\rright}{\textup{right}}
\newcommand{\lift}{\textup{lift}}
\newcommand{\lsize}{\textup{leaf-size}}
\newcommand{\sm}{\textup{small}}
\newcommand{\rt}{\textup{root}}
\newcommand{\ver}{\textup{ver}}
\newcommand{\poly}{\textup{poly}}
\newcommand{\bottom}{\textup{bottom}}
\newcommand{\shape}{\textup{shape}}
\newcommand{\graph}{\textup{graph}}
\newcommand{\rev}{\textup{rev}}

\newcommand{\fra}[3]{\left(\frac{#1}{#2}\right)^{#3}}

\newcommand{\eps}{\varepsilon}

\newcommand{\Nplus}{\N_{\geq 0}}	
\newcommand{\zeroel}[1]{0_{#1}}
\newcommand{\floor}[1]{\lfloor #1 \rfloor}
\newcommand{\ceil}[1]{\lceil #1 \rceil}
\newcommand{\under}{\sqsubseteq}
\newcommand{\ut}{\mbox{UT}}

\renewcommand{\enspace}{}

\newcommand{\nl}{\textup{NL}\xspace}
\newcommand{\nctwo}{\textup{NC$^2$}\xspace}
\newcommand{\ptime}{\textup{PTime}\xspace}
\newcommand{\np}{\textup{NP}\xspace}
\newcommand{\pspace}{\textup{PSpace}\xspace}
\newcommand{\exptime}{\textup{ExpTime}\xspace}
\newcommand{\expspace}{\textup{ExpSpace}\xspace}
\newcommand{\twoexptime}{\textup{2-ExpTime}\xspace}
\newcommand{\twoexpspace}{\textup{2-ExpSpace}\xspace}
\newcommand{\tower}{\textup{Tower}\xspace}
\newcommand{\ackermann}{\textup{Ackermann}\xspace}
\newcommand{\hypackermann}{\textup{HyperAckermann}\xspace}

\newcommand{\myparagraph}[1]{\vskip 0.3cm\textbf{#1.}}

\newcommand{\goto}[2]{\textbf{goto} {\footnotesize #1} \textbf{or} {\footnotesize #2}}
\newcommand{\gotod}[1]{\textbf{goto} {\footnotesize #1}}
\newcommand{\zerotest}[1]{\textbf{zero-test}($#1$)}
\newcommand{\testm}[1]{\textbf{max?}~$#1$}
\newcommand{\inc}[1]{\add{#1}{1}}
\newcommand{\dec}[1]{\sub{#1}{1}}
\newcommand{\add}[2]{$\coreadd{\vr{#1}}{#2}$}
\newcommand{\sub}[2]{$\coresub{\vr{#1}}{#2}$}
\newcommand{\coreadd}[2]{#1 \,\, +\!\!= \, #2}
\newcommand{\coresub}[2]{#1 \,\, -\!\!= \, #2}
\newcommand{\vr}[1]{#1}
\newcommand{\halt}{\textbf{halt}}
\newcommand{\haltz}[1]{{\halt} \textbf{if} $#1 = 0$}
\newcommand{\initialise}{\textbf{initialise to} $0$}

\maketitle

\begin{abstract}
We consider the model of one-dimensional Pushdown Vector Addition Systems (1-PVAS), a fundamental computational model simulating both recursive and concurrent behaviours. Our main result is decidability of the reachability problem for 1-PVAS, an important open problem investigated for at least a decade. In the algorithm we actually consider an equivalent model of Grammar Vector Addition Systems (GVAS). We prove the main result by showing that for every one-dimensional GVAS (1-GVAS) one can compute another 1-GVAS, which has the same reachability relation as the original one and additionally has the so-called thin property. Due to the work of Atig and Ganty from 2011, thin 1-GVAS have decidable reachability problem, therefore our construction implies decidability of the problem for all 1-GVAS. Moreover, we also show that if reachability in thin 1-GVAS can be decided in elementary time then also reachability in all 1-GVAS can be decided in elementary time.
\end{abstract}

\newpage


\section{Introduction}\label{sec:intro}
Automata extended with counters or with a pushdown are one of the most fundamental computational models
with countless applications in both theory and practice
(see~\cite{DBLP:journals/siglog/Schmitz16,DBLP:conf/concur/BouajjaniEM97,DBLP:journals/eatcs/EsparzaN94}).
In turn, one of the most natural decision problems
for a computational model is the emptiness problem asking whether its language is empty or, equivalently,
the reachability problem asking whether there exists a run of this model between given source and target configurations.
It is a folklore that the reachability problem becomes undecidable for automata with two pushdowns and even with
nonnegative integer counters equipped with zero-tests.
This last model is also often called the Minsky machines, as the undecidability was shown by Marvin Minsky in 1967~\cite{Minsky67}.
The proof essentially shows that
automata with two zero-tested counters are Turing powerful and the
result follows from the undecidability of the halting problem for Turing machines.

Since so seemingly simple models already have undecidable reachability problems, the community has been exploring various simpler models for decades, notably a two classical models, which in very natural ways simulate recursion and concurrency.
The first prominent example is the pushdown automaton, which is a fundamental model for simulating recursion.
The reachability problem for pushdown automata can be rather easily shown to be in polynomial time~\cite{DBLP:books/aw/HopcroftU79}
and many techniques are developed for various algorithmic problems for pushdown automata.
The second very widely investigated model is automata with nonnegative integer counters, but without the possibility of zero-testing them.
These are Vector Addition Systems with States (VASS) and is essentially equivalent to Petri nets, which are 
popular for modelling concurrent processes~\cite{DBLP:journals/jcsc/Aalst98}.
Algorithmic properties of VASS together with its central reachability problem
are studied for decades, however the complexity of the reachability problem was established only recently.
The problem is known to be decidable in \ackermann~\cite{DBLP:conf/lics/LerouxS19}
and also \ackermann-hard~\cite{DBLP:conf/focs/Leroux21,DBLP:conf/focs/CzerwinskiO21}, which implies \ackermann-completeness.
Despite the fact that the complexity of the reachability problem is settled, many related questions are still open, for example
the complexity of the problem for fixed number of counters. Thus, one may say that there is already some understanding
of how to solve algorithmic problems for VASS, a fundamental model for concurrent programs.
However, the picture is much less clear if we would like to model recursion and concurrency together.
One such model is automaton with two resources: one pushdown and many counters, it is called
Pushdown Vector Addition System (PVAS). Decidability of the reachability problem for that model is open for a long time
and has been unofficially conjectured to be decidable.
It was called a long-standing open problem already in 2011 (see the abstract of~\cite{DBLP:conf/fsttcs/AtigG11}).
The reachability problem for PVAS was mentioned and researched a lot in the last
decade~\cite{DBLP:conf/csl/LerouxPS14,DBLP:conf/birthday/LazicT17,DBLP:conf/fsttcs/Lazic19}.
Importantly, even for the one-dimensional case the coverability, reachability and related problems were investigated~\cite{DBLP:conf/rp/LerouxST15,DBLP:conf/icalp/LerouxST15,DBLP:journals/ipl/EnglertHLLLS21}
and notably, this investigation led to an important progress on the VASS reachability problem~\cite{DBLP:conf/stoc/CzerwinskiLLLM19}
(see the explanation here~\cite{DBLP:conf/fsttcs/Lazic19}).
Despite of all of that efforts, decidability of the reachability problem for one-dimensional PVAS still remains a stubborn obstacle for a progress in this area.

It seems that in automata theory, there are well-developed techniques for handling pushdown (i.e., recursion), some methods for managing multiple counters (i.e., concurrency),
but there exist no matured techniques for dealing with both pushdown and counters (or both recursion and concurrency) at the same time.
Perhaps surprisingly, there is no strong evidence that these problems are extremely hard computationally, the best
known lower complexity bound for the reachability problem for the one-dimensional PVAS is \pspace-hardness~\cite{DBLP:conf/lics/BlondinFGHM15}. The big gap in the known complexity bounds suggests that quite possibly some insightful techniques still
wait for being invented and explored, which motivates our search in this direction.

\paragraph*{Related research}
Even though decidability of the reachability problem for PVAS is still open, even in dimension one,
some progress has been made towards it.
Most of the existing results rely on the fact that the class of languages of PVAS is equal to the class
of languages of Grammar Vector Addition Systems (GVAS) and the transformations in both ways
are effective~\cite{DBLP:conf/icalp/LerouxST15}
(similarly as for languages of pushdown automata and languages of context-free grammars).
A $d$-dimensional GVAS (d-GVAS) is just a context-free grammar, where the terminals are vectors in $\Z^d$.
A derivation of a d-GVAS is a valid derivation from a vector $s \in \N^d$ to a vector $t \in \N^d$
if starting at $s$ and adding up consecutive vectors from the yield of the derivation we never go below zero
on any dimension, and at the end reach vector $t$.
The GVAS are often much more convenient than PVAS for designing algorithms,
thus majority of the results focus on this model. We also formulate our results using the terminology of GVAS.
In 2011 Atig and Ganty have proven decidability of GVAS emptiness in the special case when the input GVAS
has a finite-index~\cite{DBLP:conf/fsttcs/AtigG11}, we show that finite-index GVAS are actually equivalent to thin GVAS.
The class is quite restricted, however it contains for example
linear grammars. We build our contribution on the work of Atig and Ganty.

In~\cite{DBLP:conf/csl/LerouxPS14} Leroux, Praveen and Sutre showed that if the reachability set of a PVAS is finite than
it is of size at most hyper-Ackermannian, which immediately implies decidability of the boundedness problem,
asking whether the reachability set is finite.
Interestingly the bound is tight and the hyper-Ackermannian size of the reachability sets can indeed be obtained
as shown in the same paper~\cite{DBLP:conf/csl/LerouxPS14}.
This suggests that the reachability problem for PVAS might be \hypackermann-hard,
but no results of this kind have been obtained and the best complexity lower bound for the problem
is still the \ackermann-hardness inherited from Vector Addition Systems~\cite{DBLP:conf/focs/Leroux21,DBLP:conf/focs/CzerwinskiO21}.
Further progress on the boundedness problem
for 1-GVAS was achieved by Leroux, Sutre and Totzke in~\cite{DBLP:conf/rp/LerouxST15},
however it does improve the complexity of this problem.

In~\cite{DBLP:conf/icalp/LerouxST15} the same authors have proven that the coverability problem, asking whether
a given configuration can be covered, is decidable for 1-GVAS. Moreover, at the end of their paper, they remark
that the given algorithm works in exponential space.
Interestingly, for PVAS the reachability and the coverability problem are interreducible (to our best knowledge
it is first mentioned in~\cite{DBLP:journals/corr/Lazic13}, but without a proof), in a sheer contrast
to Vector Addition Systems, where the coverability problem is \expspace-complete~\cite{Lipton76,DBLP:journals/tcs/Rackoff78},
while the reachability problem is \ackermann-complete~\cite{DBLP:conf/lics/LerouxS19,DBLP:conf/focs/Leroux21,DBLP:conf/focs/CzerwinskiO21}. The reduction of the coverability problem for of $d$-dimensional PVAS (d-PVAS) to the reachability
problem for d-PVAS is immediate.
On the other hand the reachability problem for d-PVAS can be reduced to the coverability problem for $(d+1)$-PVAS by a simple,
but clever construction: the additional $(d+1)$-th counter is used to keep the sum
of the counters constant along the run; its initial value is guessed at the beginning of the run,
remembered at the bottom of the stack and checked at the end of the run.
Therefore, there is hierarchy of problems of increasing hardness for PVAS:
coverability for 1-PVAS, reachability for 1-PVAS, coverability for 2-PVAS, reachability for 2-PVAS, etc.
As the coverability problem is known to be decidable for 1-PVAS the next natural next step towards understanding
the reachability-type problems for PVAS seems to be solving the reachability problem for 1-PVAS,
which is the topic of this paper.

There is also a lot of research on the models related to PVAS, which does not directly motivate our direction,
but gives a broader view on the area and witnesses a vivid interest of the community in this kind of problems.
The investigated models are often a slight relaxations of PVAS,
for which decidability results about the reachability problem are known.
A special case of a pushdown is a counter, which can be zero tested. It is known
that the reachability problem is decidable for VASS with one zero-test (one counter
allowed to be zero-tested)~\cite{DBLP:journals/entcs/Reinhardt08,DBLP:conf/mfcs/Bonnet11}
and even with hierarchical zero-tests~\cite{DBLP:journals/entcs/Reinhardt08,DBLP:conf/icalp/Guttenberg24}.
In dimension two even more is known. For 2-VASS with one zero-test in~\cite{DBLP:conf/concur/LerouxS20}
it is shown that similarly as for 2-VASS the reachability relation is semilinear and the reachability problem is \pspace-complete,
thus in this case the zero-test is practically for free. The reachability problem is decidable even for 2-VASS with one zero-test
and one reset~\cite{DBLP:conf/fsttcs/FinkelLS18}.
If we allow counters of PVAS to drop below zero we obtain a model of $\Z$-PVAS. For $\Z$-PVAS it is known that
the reachability relation is semilinear~\cite{DBLP:journals/jcss/HarjuIKS02} and the reachability problem
is decidable~\cite{DBLP:journals/jcss/HarjuIKS02} and even \np-complete~\cite{DBLP:conf/cav/HagueL11}.
Another recently studied model related to PVAS is a bidirected PVAS: we demand that for each transition in a PVAS
there exists another transition, which reverses the effect of the first one. Ganardi et al. show
in~\cite{DBLP:conf/lics/GanardiMZ22} that the reachability problem in bidirected PVAS is decidable in \ackermann and is \tower-hard.
Moreover they prove that in each fixed dimension it is primitive recursive and in dimension one it is in \pspace.
However, the techniques from many above mentioned papers,
like~\cite{DBLP:conf/fsttcs/FinkelLS18, DBLP:conf/cav/HagueL11, DBLP:conf/lics/GanardiMZ22} cannot be directly of use,
as they rely on the fact that the reachability relation in 2-VASS with zero-test and reset, in $\Z$-PVAS and in bidirected
PVAS, respectively, is semilinear (in other words Presburger definable). This is unfortunately not the case already
for 1-PVAS, which is rather easy to show (see Example 2.1. in~\cite{DBLP:journals/lmcs/LerouxPSS19}).

\paragraph*{Our contribution}
The main contribution of this paper is the following theorem.

\begin{theorem}\label{thm:decidability}
The reachability problem for one-dimensional Grammar Vector Addition Systems is decidable.
\end{theorem}

The proof strongly relies on our conceptual contribution of partitioning the nonterminals of 1-GVAS
into \emph{thin} and \emph{branching} ones, and treat them differently.
This division can actually be defined for any context-free grammar,
regardless of whether the terminals are integers (as in 1-GVAS),
vectors in $\Z^d$ (as in d-GVAS), or letters from a finite alphabet.
We say that a nonterminal $X$ is branching if, intuitively, it can create two copies of itself; specifically, this means there is a derivation of the form $X \Rightarrow \alpha X \beta X \gamma$ for some sequences
$\alpha$, $\beta$ and $\gamma$ of nonterminals and terminals.
Otherwise, for any derivation $X \Rightarrow \alpha X \beta$, no nonterminal from $\alpha$ or $\beta$ can generate $X$, and $X$ is said to be thin. We call a GVAS \emph{thin} if all its nonterminals are thin.
It turns out that both thin and branching nonterminals have nice properties, but these properties
are useful in very different situations. 
Therefore the partition of nonterminals into thin and branching is very helpful for the analysis of 1-GVAS,
and potentially, may be helpful also for future research on higher dimensional GVAS.

We strongly rely on the work of Atig and Ganty, who showed in~\cite{DBLP:conf/fsttcs/AtigG11} that the reachability problem is decidable for a subclass of GVAS called finite-index GVAS. They are equivalent to thin GVAS.
Therefore, to show Theorem~\ref{thm:decidability}
it is enough to eliminate branching nonterminals from a 1-GVAS.
This is indeed possible, as shown in the core of our technical contribution, the following theorem.

\begin{theorem}\label{thm:main}
There is an algorithm, which given a 1-GVAS $G$ produces a thin 1-GVAS $H$ with the same reachability relation.
Moreover, the size of the produced thin 1-GVAS $H$ is bounded by $s = f(\size(G))$ for some triply-exponential function $f$
and the construction works in time $s \cdot g(s)$, where $g$ is the time complexity of the reachability problem for thin 1-GVAS.
\end{theorem}

It is not clear how to obtain function $f$ in Theorem~\ref{thm:main} smaller than triply-exponential (see Remark~\ref{rem:triply-exp}).
Therefore, even if the complexity of reachability for thin 1-GVAS is very low we cannot get a very low complexity for 1-GVAS
using Theorem~\ref{thm:main} directly. However, it is natural to conjecture that the complexity of reachability for 1-GVAS is elementary.
We leave this problem open, but we notice that Theorem~\ref{thm:main} immediately implies the following corollary,
which makes designing an elementary algorithm for 1-GVAS easier.

\begin{corollary}\label{cor:thin-elementary}
If the complexity of the reachability problem for thin 1-GVAS is elementary
then the complexity of the reachability problem for 1-GVAS is elementary (without the thinness assumption).
\end{corollary}

The rest of the paper is organised as follows.
In Section~\ref{sec:prelim} we introduce preliminary notions.
Then, in Section~\ref{sec:overview} we give an intuitive overview on the main ideas behind our contribution,
formulate key Lemmas~\ref{lem:far-from-axis},~\ref{lem:small-lines}~and~\ref{lem:bounded-area} and show
how they imply Theorem~\ref{thm:main}. In the following sections we prove the auxiliary lemmas.
In Section~\ref{sec:far-from-axis} we show Lemma~\ref{lem:far-from-axis}.
Next, in Section~\ref{sec:mainproof} we provide the most challenging proof of Lemma~\ref{lem:small-lines}.
Further, in Section~\ref{sec:bounded-area} we prove Lemma~\ref{lem:bounded-area}.
Finally, in Section~\ref{sec:future} we discuss possible interesting future research directions.


\section{Preliminaries}\label{sec:prelim}
For $a, b \in \Z$ by $[a, b]$ we denote the set $\{k \in \Z \mid a \leq k \leq b\}$.
For a vector $v \in \Z^d$ and $i \in [1,d]$ by $v[i]$ we denote the $i$-th entry of $v$.
For a relation $R \subseteq \N^2$ and $a \in \N$ by $R(a)$ we denote the set $\{b \mid R(a,b)\}$.
A set $S \subseteq \N^d$ is \emph{linear} if it is of a form $b + P^*$ for some $b \in \N^d$ and $P \subseteq_\fin \N^d$.
In other words $S = \set{b + k_1 p_1 + \ldots + k_m p_m \mid k_1, \ldots, k_m \in \N}$, where $P = \set{p_1, \ldots, p_m}$.
A set $S \subseteq \N^d$ is \emph{semilinear} if it is a finite union of linear sets.

\paragraph*{GVAS and derivations}
A \emph{d-dimensional Grammar Vector Addition System} (shortly, a d-GVAS or a GVAS if $d$ is irrelevant)
is briefly speaking a context-free grammar with terminals being vectors in $\Z^d$.
More precisely, a d-GVAS $G$ consists of a finite set of \emph{nonterminals} $\Nn(G)$ with a distinguished \emph{initial nonterminal} $S \in \Nn(G)$ and a finite set of \emph{production rules} $\Pp(G)$ of a form $R \subseteq \Nn(G) \times (\Nn(G) \cup \Z^d)^*$.
A rule $(X, \alpha) \in \Nn(G) \times (\Nn(G) \cup \Z^d)^*$ is often written as $X \to \alpha$.
The initial nonterminal of $G$ is denoted by $\S(G)$.
For a nonterminal $X \in \Nn(G)$ by $G_X$ we define $G$ with the initial nonterminal set to $X$.

The \emph{size} of a number $n \in \Z$ is the number of bits in the binary encoding of its absolute value.
The \emph{size} of a vector $v \in \Z^d$ is the sum of sizes of its entries.
The \emph{size} of $G$, denoted $\size(G)$, is the sum of cardinalities of the sets of nonterminals $\Nn(S)$,
rules $\Pp(G)$, length of all the right hand sides of the rules in $\Pp(G)$ and sizes of all the vectors occurring
in the right hand sides of the rules in $\Pp(G)$.

In the sequel we assume that each rule in $\Pp(G)$ is of the form $X \to Y Z$,
where $X$ is a nonterminal, while $Y$ and $Z$ are either nonterminals or terminals.
This assumption is useful in the sequel.
And it is harmless, as each GVAS can be transformed into an equivalent GVAS satisfying this assumption,
with at most polynomial blowup in size.
Indeed, if there are some rules, such that the right-hand side is shorter than two (so is either empty or contains
single nonterminal or terminal) then we can add to these right-hand sides additional terminals being the zero vectors. 
Therefore it is enough to show how we eliminate rules of longer right-hand sides.
The main idea is that any rule $X \to X_1 \cdots X_k$ for $k \geq 3$
can be simulated by adding nonterminals $Y_2, Y_3, \ldots, Y_{k-1}$ and creating the following rules:
$X \to Y_{k-1} X_k$, $Y_i \to Y_{i-1} X_i$ for $i \geq 3$ and $Y_2 \to X_1 X_2$.
Sometimes, for brevity, we use in examples rules with longer right-hand sides, but we do this only for illustrative reasons
and the GVAS we present can be transformed to binary ones using the above mentioned technique.

\paragraph*{Derivations, cycles and effects}
For a given d-GVAS $G$ and its nonterminal $X \in \Nn(G)$ an \emph{$X$-derivation} 
is an ordered tree (namely the order of children of a given node matters)
with nodes labelled by nonterminals or vectors from $\Z^d$ in such a way that:
\begin{itemize}
  \item the root is labelled by the nonterminal $X$, and
  \item for each internal node $v$ there is a rule $Y \to \alpha$ such that $v$ is labelled by the nonterminal $Y$
  and the sequence of labels of the children of $v$ is exactly $\alpha$.
\end{itemize}
A \emph{derivation} is an $X$-derivation for some $X \in \Nn(G)$.
Notice that only leaves can be labelled by vectors from $\Z^d$, but we do not require all the leaves 
of a derivation to be labelled by vectors, some can be labelled by nonterminals.
A \emph{yield} of a derivation is the word obtained from labels of its leaves read from the left to the right.
A derivation is \emph{complete} if no leaf is labelled by a nonterminal, so all leaves are labelled by vectors.
A derivation is \emph{simple} if no path contains two nodes labelled by the same nonterminal.
The \emph{size} of a derivation is the total number of vertices in this derivation.
We call a binary tree \emph{full} if each its node is either a leaf or has exactly two children.
Notice that derivations of the GVAS we consider are always full binary trees.

For an $X$-derivation $\tau$ of yield $\alpha$ we often write $X \Rightarrow \alpha$.
We often call $X \Rightarrow \alpha$ also a derivation, even though there can be many $X$-derivations
with the yield $\alpha$. By $X \Rightarrow \alpha$ we mean any of those $X$-derivations, for example $\tau$.

For a nonterminal $X \in \Nn(G)$ an \emph{$X$-cycle} is an $X$-derivation with a distinguished leaf labelled
by the same nonterminal $X$. A \emph{cycle} is an $X$-cycle for some nonterminal $X \in \Nn(G)$.
A cycle is \emph{complete} if only the distinguished leaf is labelled by a nonterminal.
A cycle is \emph{simple} if after removing the distinguished leaf it is a simple derivation, so does not contain
any path with two nodes labelled by the same nonterminal.

For a cycle, its \emph{left effect} is the sum of all the vectors labelling its leaves to the left of the distinguished leaf
and its \emph{right effect} is the sum of all the vectors labelling its leaves to the right of the distinguished leaf.
The \emph{global effect} of a cycle is the sum of its left effect and right effect.
The \emph{effect} of a complete derivation is the sum of all the vectors in its leaves.
Notice that cycles are not complete derivations, thus they do not have a well defined effect, but they have a global effect instead.

\paragraph*{Graph of nonterminals}
For any context-free grammar $G$ by the \emph{graph of $G$} we mean the graph
with
\begin{itemize}
  \item vertices being nonterminals $\Nn(G)$ of $G$, and
  \item edges $X \to Y$ if there is a rule $X \to \alpha$ in $\Pp(G)$ such that $Y \in \alpha$.
\end{itemize}
Intuitively, an edge $X \to Y$ means that $Y$ can be directly derived from $X$ and a path
from $X$ to $Y$ in the graph of $G$ means that $Y$ be be derived from $X$ by some sequence of rules.
We sometimes also say in such a situation that $X$ \emph{reaches} or \emph{derives} $Y$.
A \emph{component of $G$} is the strongly connected component of the graph of $G$. We slightly overload
the notation and often use the term component to denote the set of nonterminals belonging to a given component.
The \emph{dag of components of $G$} is the directed acyclic graph of the components of the graph of $G$,
where the component of $X$ has an edge to a component of $Y$ if $X$ reaches $Y$.
The \emph{top component}, denoted $\topp(G)$, is the component, which contains the initial nonterminal $\S(G)$.
A nonterminal $X$ is a \emph{top nonterminal} if $X \in \topp(G)$. Otherwise, it is a \emph{lower nonterminal}.
The set of all the lower nonterminals of $G$ is denoted by $\low(G)$.

\paragraph*{Thin and branching nonterminals}
We say that a nonterminal $X$ is \emph{branching} if there is a derivation $X \To \alpha X \beta X \gamma$
for some sequences of nonterminals and terminals $\alpha$, $\beta$ and $\gamma$.
Intuitively, a branching nonterminal can derive two copies of itself simultaneously.
Otherwise we say that $X$ is \emph{thin}. Notice the following claim.

\begin{claim}\label{cl:reachable-branching}
If $X$ and $Y$ are in the same component and $X$ is branching then $Y$ is also branching.
\end{claim}

\begin{proof}
This is because we have derivations $X \To \alpha X \beta X \gamma$, $X \To \delta_1 Y \delta_2$, $Y \To \delta_3 X \delta_4$,
so also
\[
Y \To \delta_3 X \delta_4 \To \delta_3 \alpha X \beta X \gamma \delta_4 \To \delta_3 \alpha \delta_1 Y \delta_2 \beta \delta_1 Y \delta_2 \gamma \delta_4.
\]
Therefore, indeed $Y$ is branching as well.
\end{proof}
By Claim~\ref{cl:reachable-branching} for each component either all the nonterminals in it are branching or all are thin.
Therefore we call a component \emph{branching} if the nonterminals in it are branching and \emph{thin} otherwise.

Notice that if $X$ is thin, $X \To \alpha Y \beta$ and $Y$ is in the same component as $X$, then necessarily
all the nonterminals in $\alpha$ and $\beta$ cannot derive $X$, they are in the other components.
The intuition is that a thin nonterminal can only produce one copy of itself, all the other nonterminals
should be from lower components (i.e. components lower in the dag of component).
If all the nonterminals in a context-free grammar $G$ are thin then we say that \emph{$G$ is thin}.

Atig and Ganty in~\cite{DBLP:conf/fsttcs/AtigG11} defined finite-index grammars and showed
that the reachability problem is decidable for them. Their result is important for us,
so we need to argue that each thin grammar is also a finite-index one. 
We write here $\alpha X \beta \to \alpha \delta \beta$ for any rule $X \to \delta$ in $\Pp(G)$.
In other words $\gamma \to \gamma'$ if $\gamma'$ is obtained from $\gamma$ by taking one
nonterminal $X$ from $\gamma$ and substituting it by the right-hand side of a rule $X \to \delta$.
Then, we say that a grammar $G$ is \emph{finite-index} if there exists a $k \in \N$
such that for each word over terminals $\alpha$ derived from $\S(G)$ there is a sequence
\[
\S(G) \to \alpha_1 \to \alpha_2 \to \ldots \to \alpha_m = \alpha
\]
such that for each $i \in [1,m]$ the word $\alpha_i$ contains at most $k$ occurrences of nonterminals.
\begin{claim}\label{cl:finite-index-thin}
Every thin context-free grammar is finite-index.
\end{claim}

For the smoothness of reading we put the proof of Claim~\ref{cl:finite-index-thin}
(which is easy and rather straightforward) in Section~\ref{sec:app1}.

\paragraph*{Reachability relations of a GVAS}
For a d-GVAS $G$, its nonterminal $X$ and two vectors $u, v \in \N^d$ we say that $u$ \emph{$X$-reaches} $v$
if there is a complete $X$-derivation $\tau$ with the yield $v_1, \ldots, v_k$, where $v_i \in \Z^d$
for all $i \in [1,k]$, such that for each $j \in [1,k]$ we have $u + v_1 + \ldots + v_j \in \N^d$
and additionally $u + v_1 + \ldots + v_k = v$. Intuitively, if we start in vector $u \in \N$, traverse the leaves
of a derivation $\tau$ from left to right and add their labels (in $\Z^d$) to our current vector, then we never drop counter values below
zero (so always stay in $\N^d$) and at the end of the traversal we reach vector $v \in \N^d$.
We write then $u \trans{X} v$ or $u \trans{G} v$ if $X$ is the initial nonterminal of $G$.

The \emph{reachability relation of $G$}, denoted $\reach{G}$, is the set of pairs $(u, v) \in \N^d \times \N^d$
such that $u \trans{G} v$. Notice that for each d-GVAS $G$ the reachability relation $\reach{G}$ is \emph{diagonal}, so for each $\Delta \in \N^d$ and $u, v \in \N^d$, if $(u, v) \in \reach{G}$ then also $(u + \Delta, v + \Delta) \in \reach{G}$. Intuitively, this is because if you shift counters
by $\Delta \in \N^d$ they will only be higher, so that operation cannot cause dropping below zero.

For 1-GVAS $G$ the reachability relation $\reach{G}$ is a subset of $\N^2$.
Then for a fixed value $a \in \N$ one can see $\reach{G}(a)$ (defined above as $\{b \mid \reach{G}(a,b)\}$) to be the reachability set
of $a$. Notice also that $\set{(a,b) \mid \reach{G}(a,b)}$ is a subset of a vertical line in $\N^2$, so often informally refer to $\reach{G}(a)$
as a vertical line.

The following notion plays an important role in our proof: for a set $S \subseteq \N^2$ and $1$-GVAS $G$, $H$ we say that $G$ \emph{$S$-exactly under-approximates} $H$ if the following two points hold:
(1) $\reach{G} \subseteq \reach{H}$, and (2) $\reach{G} \cap S = \reach{H} \cap S$.
We write then $G \under_S H$.

\paragraph*{Reachability and coverability problems}
The \emph{reachability problem for d-GVAS} asks, given d-GVAS $G$ and two vectors $u, v \in \N^d$,
whether $u \trans{G} v$. The \emph{coverability problem for d-GVAS} asks, given d-GVAS $G$ and two vectors $u, v \in \N^d$
whether there exists $v' \in \N^d$ such that $u \trans{G} v'$ and for each $i \in [1,d]$ we have $v[i] \leq v'[i]$.

\begin{wrapfigure}{l}{0.3 \textwidth}
\centering
\begin{tikzpicture}[
    sibling distance=20mm, level distance=10mm]

\node[circle,fill, minimum size=1mm, inner sep=0pt] (root) {} 
    child {
        node[circle,fill, minimum size=1mm, inner sep=0pt] (left) {} 
            child { node[circle,fill, minimum size=1mm, inner sep=0pt] (leftleft) {} } 
            child { node[circle,fill, minimum size=1mm, inner sep=0pt] (leftright) {} } 
    }
    child {
        node[circle,fill, minimum size=1mm, inner sep=0pt] (right) {} 
    };

\node[left=2mm of root] (1) {};
\node[right=2mm of root] (10) {};

\node[left=2mm of left] (2) {};
\node[right=2mm of left] (7) {};

\node[left=2mm of right] (8) {};
\node[right=2mm of right] (9) {};

\node[left=2mm of leftleft] (3) {};
\node[right=2mm of leftleft] (4) {};

\node[left=2mm of leftright] (5) {};
\node[right=2mm of leftright] (6) {};

\draw[red, thick] (1.west) to (3.west);

\draw[red, thick] (3.west) arc[start angle=180, end angle=360, radius=0.52cm];
\draw[red, thick] (5.west) arc[start angle=180, end angle=360, radius=0.52cm];
\draw[red, thick] (8.west) arc[start angle=180, end angle=360, radius=0.52cm];

\draw[red, thick] (4.east) arc[start angle=180, end angle=0, radius=0.47cm];
\draw[red, thick] (7.east) arc[start angle=180, end angle=0, radius=0.47cm];

\draw[red, thick] (6.east) to (7.east);

\draw[red, thick, ->] (9.east) to (10.east);

\end{tikzpicture}
\caption{Direction of the flow.}
\label{fig:flow-direction}
\end{wrapfigure}

\paragraph*{Derivations with counters}
We assume now that we work with derivations of 1-GVAS, not a general GVAS.
Therefore, leaves of derivations are labelled by nonterminals or by integer values.
An \emph{$X$-derivation with counters} is a complete $X$-derivation in which every node is additionally labelled with
two nonnegative integer numbers: \emph{input} and \emph{output}. We say simply \emph{derivation with counters} if $X$ is irrelevant.
A derivation with counters is \emph{valid} if, intuitively, the input and output counters correspond to a correct trip around the tree,
in the order illustrated on Figure~\ref{fig:flow-direction}.

\begin{wrapfigure}{r}{0.45 \textwidth}
\centering
\begin{tikzpicture}[
    sibling distance=30mm, level distance=10mm]
\node[circle, minimum size=1mm, inner sep=0pt] (root) {X} 
    child {
        node[circle, minimum size=1mm, inner sep=0pt] (left) {Y} 
            child { node[circle, minimum size=1mm, inner sep=0pt] (leftleft) {-10} } 
            child { node[circle, minimum size=1mm, inner sep=0pt] (leftright) {12} } 
    }
    child {
        node[circle, minimum size=1mm, inner sep=0pt] (right) {-7} 
    };
\node[red, left=2mm of root] (1) {20};
\node[red, right=2mm of root] (10) {15};
\node[red, left=2mm of left] (2) {20};
\node[red, right=2mm of left] (7) {22};
\node[red, left=2mm of right] (8) {22};
\node[red, right=2mm of right] (9) {15};
\node[red, left=2mm of leftleft] (3) {20};
\node[red, right=2mm of leftleft] (4) {10};
\node[red, left=2mm of leftright] (5) {10};
\node[red, right=2mm of leftright] (6) {22};
\end{tikzpicture}
\caption{Derivation with counters.}
\label{fig:deriv-counters}
\end{wrapfigure}
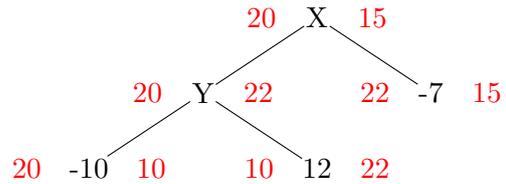

More precisely, a derivation with counters if valid, if the following four
conditions, called \emph{flow conditions}, are satisfied: (1) for each node with children, the input of the left child equals the input of the parent, (2) for each node with children, the output of the right child equals the output of the parent, (3) for each node with children, the output of the left child equals the input of the right child, and (4) for each leaf, its output is the sum of its input and its label.

It is easy to notice that for a 1-GVAS $G$ and $a, b \in \N$ there is a valid derivation with counters with the root having input $a$ and output $b$ if and only if $a \trans{G} b$. On Figure~\ref{fig:deriv-counters} we present an example of a derivation with counters witnessing
$20 \trans{X} 15$, the input and output labels of the nodes are drawn in red.

For a complete derivation $\tau$ of 1-GVAS we sometimes call its yield a \emph{run},
then a run is just a sequence of integers: $a_1, \ldots, a_k$.
We say that such a run is \emph{valid at $n$}, for $n \in \N$ if for each $i \in [1,k]$
we have $n + a_1 + \ldots + a_i \geq 0$. Notice that in such a case we can label the derivation $\tau$
by input and output values and get a derivation tree with counters, with input of the root equal $n$
and output of the root equal $n + a_1 + \ldots + a_k$. For example the derivation illustrated on Figure~\ref{fig:deriv-counters}
defines a run $-10 \hskip 0.15cm 12 \hskip 0.15cm -7$. It is valid at $20$, as witnessed at the picture, but it is also valid at $10$.
It is however not valid at $9$ or lower value.

\paragraph*{Reverse of a GVAS}
For a GVAS $G$ we define its \emph{reverse}, denoted $G_\rev$, as follows: it has all the same nonterminals,
but for a rule $X \to \alpha \beta$ in $\Pp(G)$ we add the rule $X \to \beta' \alpha'$ to $\Pp(G_\rev)$,
where $\alpha' = \alpha$ if $\alpha$ is a nonterminal and $\alpha' = - \alpha$ if $\alpha$ is a terminal,
similarly for $\beta'$. For example if $X \to 5 Y \in \Pp(G)$ then $X \to Y -5 \in \Pp(G_\rev)$.
It is easy to see the following claim.

\begin{claim}\label{cl:reverse}
Let $G$ be a 1-GVAS and $G_\rev$ be its reverse. Then for any $s, t \in \N$
it holds $(s, t) \in \reach{G}$ if and only if $(t, s) \in \reach{G_\rev}$.
\end{claim}

\paragraph*{Functions and complexity}
A function $f: \N \to \N$ is \emph{exponential} (or \emph{$1$-fold exponential}) if $f(n) \leq 2^{P(n)}$ for some polynomial $P: \N \to \N$.
A function $f: \N \to \N$ is \emph{$(k+1)$-fold exponential} if $f(n) \leq 2^{g(n)}$ for some $k$-fold exponential function $g: \N \to \N$.
We call $2$-fold exponential functions \emph{doubly-exponential}, $3$-fold exponential functions \emph{triply-exponential}, etc.
A function is \emph{elementary} if it is $k$-fold exponential for some $k \in \N$.

We define the fast growing hierarchy of functions as follows. We have $F_1(n) = 2n$
and $F_{k+1}(n) = \underbrace{F_k \circ \ldots \circ F_k}_{n \text{ times }}(1)$ for $k \geq 1$.
Thus, in particular, $F_2(n) = 2^n$ and $F_3(n) = \underbrace{2^{2^{\iddots^2}}}_{n \text{ times }}$.
We define $\tower(n) = F_3(n)$ and $\ackermann(n) = F_n(n)$.

For the above defined families of functions there are naturally defined complexity classes. For example, the complexity of a problem
is \emph{elementary} if it can be solved in some elementary time. There are similar definitions for the complexity classes \tower
and \ackermann; for more details see~\cite{DBLP:journals/toct/Schmitz16}.


\section{Overview}\label{sec:overview}
Recall that the main contribution of this paper is Theorem~\ref{thm:main}, which in particular implies Theorem~\ref{thm:decidability}.
Below we provide an intuitive guide through its proof. We say that a 1-GVAS is \emph{top-branching} if the component of the
initial nonterminal is branching. A 1-GVAS is \emph{top-only-branching} if the component of the initial nonterminal is branching,
but all the other are thin. First notice that the following lemma would be sufficient to prove that for each 1-GVAS one can
compute an equivalent thin 1-GVAS.

\begin{lemma}\label{lem:top-only-branching}
Let $G$ be a top-only-branching 1-GVAS.
Then one can compute a thin 1-GVAS $H$ equivalent to $G$.
\end{lemma}

Indeed, to transform any 1-GVAS $G$ to a thin 1-GVAS we look at the dag of components of $G$.
We choose any component $C$ which is branching, but all the components below it are thin.
For every nonterminal $V$ in the component $C$ the 1-GVAS $G_V$ is top-only-branching.
Therefore by Lemma~\ref{lem:top-only-branching} we can compute thin 1-GVAS $H_V$ equivalent to $G_V$.
Next, by substituting in $G$ in each rule containing any nonterminal $V$ from the component $C$
by the initial nonterminal from $H_V$ we decrease the amount of branching nonterminals in $G$.
Thus, continuing this way, at some point there will be no more branching components and in that way
we compute a thin 1-GVAS $H$ equivalent to $G$.

The above sketch illustrates well our approach. However, we need to make it more subtle, as Theorem~\ref{thm:main}
not only states computability of an equivalent thin 1-GVAS. It also claims that under the condition that reachability for
thin 1-GVAS works in elementary time the computation of the equivalent thin 1-GVAS also works in elementary time
(and in consequence also produces a thin 1-GVAS of elementary size). Therefore our lemma also needs to make statements
about the complexity. We will see in a moment, that we need to make it in a more subtle way. Assume for a moment
that in Lemma~\ref{lem:top-only-branching} we claim that the computation of $H$ takes at most exponential
time and therefore $H$ is of size at most exponential wrt. size of $G$. If our 1-GVAS $G$ has $k$ components
forming a path (one below the other) and each is branching then, while transforming $G$ into equivalent thin 1-GVAS $H$,
the exponential blowup will be nested $k$ times. Therefore, this way of producing $H$ may result in a non-elementary algorithm,
of complexity around \tower. For that reason, we need to formulate the following stronger and more subtle lemma, which takes
as input not a top-only-branching 1-GVAS, but a top-branching 1-GVAS.

\begin{lemma}\label{lem:top-branching}
Let $G$ be a top-branching 1-GVAS. 
Suppose that for each $V \in \low(G)$ we have a thin 1-GVAS $H_V$ of same reachability relation.
Then can compute a thin 1-GVAS $H$ equivalent to $G$.
Moreover, there is a triply-exponential function $f$ such that the size of $H$ is at most
$s := f(\size(G)) + \size(G) \cdot \sum_{V \in \low(G)}\size(H_V)$
and the construction works in time at most $s \cdot g(s)$ where $g$ is the time complexity of reachability for thin 1-GVAS.
\end{lemma}

\begin{remark}\label{rem:triply-exp}
The triply-exponential blowup in Lemma~\ref{lem:top-branching} comes from two sources.
The first one is because the function $f_1$ in Lemma~\ref{lem:small-lines} is triply-exponential,
the algorithm there needs to search through a tree of a doubly-exponential depth.
The second reason is that the threshold $T$ in Lemma~\ref{lem:small-lines} is exponential.
Therefore, in the construction in Section~\ref{sec:bounded-area} for each derivation of exponential depth
we add a rule to our thin 1-GVAS with the yield of that tree on the right hand side.
As there might be triply-exponentially many derivations of an exponential depth this is another reason for the triply-exponential blowup.
\end{remark}

Before we delve into the proof of Lemma~\ref{lem:top-branching} let us show how it implies Theorem~\ref{thm:main}.

\begin{proof}[Proof of Theorem~\ref{thm:main}]
To show Theorem~\ref{thm:main} assume we have a 1-GVAS $G$. Let $G$ have $n$ nonterminals
and let an algorithm for the reachability problem in thin 1-GVAS work in time complexity $g$.
Notice first that Lemma~\ref{lem:top-branching} clearly holds also for 1-GVAS, which are not top-branching.
Indeed, once we have thin 1-GVAS $H_V$ for each nonterminal $V \in \low(G)$ there is nothing to do, as $G$ with $V$
substituted by initial nonterminals of $H_V$ is then already thin and its size is at most
$\size(G) + \size(G) \cdot \sum_{V \in \low(G)}\size(H_V)$.
Therefore in the rest of the proof we assume that the considered 1-GVAS is top-branching, as this is the challenging case.

We show by induction on $n$ that the size of produced thin 1-GVAS is at most $s = (\size(G) \cdot n)^{2n} \cdot f(\size(G))$,
where $f$ is the triply-exponential function from Lemma~\ref{lem:top-branching} and the construction
works in time at most $s \cdot g(s)$. Fix $S = \size(G)$.
For the induction base, when $n = 1$, there is only one nonterminal and $\low(G)$ is empty.
Then Lemma~\ref{lem:top-branching} directly shows the induction base.
For the induction step assume the induction assumption for $n$ and show
it for $n+1$. For every $V \in \low(G)$ the 1-GVAS $G_V$ has at most $n$ nonterminals.
Therefore, by induction assumption, the size of an equivalent thin 1-GVAS $H_V$ is at most $s = (S \cdot n)^{2n} \cdot f(S)$
and the time for computing all of them is at most $n \cdot s \cdot g(s)$.
Let $s' = ((S \cdot n)^{2n+1} + 1) \cdot f(S)$ and $s'' = (S \cdot n+1)^{2n+2} \cdot f(S)$.
Thus, by Lemma~\ref{lem:top-branching} the size of $H$ is at most
\[
f(S) + S \cdot n \cdot s = f(S) + S \cdot n \cdot (S \cdot n)^{2n} \cdot f(S) = ((S \cdot n)^{2n+1} + 1) \cdot f(S) = s' \leq s'',
\]
as needed. Also, by Lemma~\ref{lem:top-branching} the time needed for computation of $H$ from all $H_V$ is at most $s' \cdot g(s')$. Together with the time needed to computing all $H_V$,
namely at most $n \cdot s \cdot g(s)$ it is at most
\[
s' \cdot g(s') + n \cdot s \cdot g(s) \leq 2s' \cdot g(s') \leq s'' \cdot g(s''),
\]
as required. The last inequality above follows from the inequality $2((S \cdot n)^{2n+1} + 1) \leq (S \cdot n +1)^{2n+2}$.
It follows from an easy inequality $2(a^b + 1) \leq (a+1)^{b+1}$ for $a = S \cdot n$ and $b = 2n+1$.
This last inequality holds as $2(a^b + 1) \leq (a+1) \cdot (a^b+1) \leq (a+1) \cdot (a+1)^b = (a+1)^{b+1}$.
It is easy to see that $s''$ is triply-exponential in $\size(G)$, as $f$ is triply-exponential, which finishes the proof.
\end{proof}

In the rest of this section we focus on the proof of Lemma~\ref{lem:top-branching}.
We assume that the considered 1-GVAS $G$ is top-branching.
Our approach is to compute thin 1-GVAS $H_1, \ldots, H_k$ such that $H_i$ $S_i$-exactly under-approximate $G$
for some sets $S_i \subseteq \N^2$. Moreover, the union of sets $S_i$ is supposed to be almost the whole quadrant $\N^2$,
concretely speaking it should be $S = \N^2 \setminus [0,B)^2$ for a computable bound $B \in \N$.
Then, for each nonterminal $X \in \topp(G)$, we can create a thin 1-GVAS $H_X$, which $S$-exactly under-approximates $G$
by taking the union of $H_i$. Finally, we will apply Lemma~\ref{lem:bounded-area} to finish the argument.

We are ready to formulate a lemma responsible for the area $[B, \infty)^2$.

\begin{lemma}\label{lem:far-from-axis}
Let $G$ be a top-branching 1-GVAS.
Then one can compute in exponential time an exponential bound $B \in \N$
and a thin 1-GVAS $H$ which $S$-exactly under-approximates $G$ for $S = [B, \infty)^2$.
\end{lemma}

The proof of Lemma~\ref{lem:far-from-axis} is presented in Section~\ref{sec:far-from-axis}.
Lemma~\ref{lem:far-from-axis} already covers a lot of the $\N^2$ in the reachability relation $R$.
It only remains to cover vertical and horizontal lines close to the axes and a bounded area close to the origin.
The following lemma takes care of the vertical lines close to the vertical axis.

\begin{lemma}\label{lem:small-lines}
Let $G$ be a top-branching 1-GVAS. 
Suppose that for each $V \in \low(G)$ we have a thin 1-GVAS $H_V$ of same reachability relation.
Then one can compute a threshold $T$ and a thin 1-GVAS $H$
which $S$-exactly under-approximates $G$ for $S = \set{a} \times [T, \infty)$.
Moreover, the threshold $T$ is at most exponential in $\size(G)$, there is a triply-exponential function $f_1$ such that
the size of $H$ is at most $s := f_1(\size(G)) + \sum_{V \in \low(G)}\size(H_V)$
and the construction works in time at most $s \cdot g(s)$ where $g$ is the time complexity of reachability for thin 1-GVAS.
\end{lemma}

It is easy to see that by Lemma~\ref{lem:small-lines} we can also take care about horizontal lines and construct
thin 1-GVAS $H$ which $S$-exactly under-approximates $G$ for $S = [T, \infty) \times \set{a}$. To obtain it,
we need to consider the reverse of $G$, namely $G_\rev$. Recall that by Claim~\ref{cl:reverse} we have $(s, t) \in \reach(G)$
if and only if $(t, s) \in \reach(G_\rev)$. Thus by applying Lemma~\ref{lem:small-lines} to $G_\rev$ we get an analog
of Lemma~\ref{lem:small-lines} for $G$, but for sets $S = [T, \infty) \times \set{a}$.

Therefore, by Lemmas~\ref{lem:far-from-axis}~and~\ref{lem:small-lines} we can compute a bound $B'$ being the maximum
of $B$ from Lemma~\ref{lem:far-from-axis} and thresholds $T$ from Lemma~\ref{lem:small-lines} such that
for each $X \in \topp(G)$ we have a thin 1-GVAS $H_X$ which $S$-exactly under-approximates $G$ for $S = \N^2 \setminus [0,B']^2$.
Notice that $B'$ is bounded by some exponential function of size of $G$, call this function $h_1$. Moreover, the size of each $H_X$
is bounded by $s := f_1(\size(G)) + \sum_{V \in \low(G)}\size(H_V)$ for some triply-exponential function $f_1$
and the construction takes at most exponential time from Lemma~\ref{lem:far-from-axis} plus at most $2 \cdot (B+1) \cdot \size(G)$ invocations of Lemma~\ref{lem:small-lines}.
This is altogether at most $h_2(\size(G)) \cdot g(s)$ time for some exponential function $h_2$
and $g$ being the time complexity of thin 1-GVAS reachability.
To finish the argument we formulate the following lemma.

\begin{lemma}\label{lem:bounded-area}
Let $G$ be a top-branching 1-GVAS and $B \in \N$.
Suppose that for each $V \in \low(G)$ we have a thin 1-GVAS $H_V$ of same reachability relation.
Suppose also that for each $X \in \topp(G)$ we have a thin 1-GVAS $H_X$ which $S$-exactly
under-approximates $G_X$ for $S = \N^2 \setminus [0,B]^2$.
Then one can compute a thin 1-GVAS $H$ equivalent to $G$.
Moreover, there is a doubly-exponential function $f_2$ such that
the size of $H$ and the time of its construction are at most 
$s := f_2(\size(G) + B) + \sum_{V \in \low(G)}\size(H_V) + \sum_{X \in \topp(G)}\size(H_X)$.
\end{lemma}

The proof of Lemma~\ref{lem:bounded-area} is shown in Section~\ref{sec:bounded-area}.
By Lemma~\ref{lem:bounded-area} applied to $G$ and constant $B'$
we can compute a thin 1-GVAS $H$ equivalent to $G$, as needed in Lemma~\ref{lem:top-branching}.
We need to show the bounds on its size and on the computation time, as required in Lemma~\ref{lem:top-branching}.

The size of $H$ is bounded by $f_2(\size(G) + B') + \sum_{V \in \low(G)}\size(H_V) + \sum_{X \in \topp(G)}\size(H_X)$.
As the size of $H_X$ is at most $f_1(\size(G)) + \sum_{V \in \low(G)}\size(H_V)$ for each $X$ then
\begin{align*}
\size(H) & \leq s_2 :=  f_2(\size(G) + B') + \sum_{V \in \low(G)}\size(H_V) \\
& + (\size(G)-1) \cdot (f_1(\size(G)) + \sum_{V \in \low(G)}\size(H_V)) \\
& \leq h_3(\size(G)) + \size(G) \cdot \sum_{V \in \low(G)}\size(H_V)
\end{align*}
for some triply-exponential function $h_3$. Indeed $h_3(n) \leq f_2(n + h_1(n))) + (n-1) \cdot f_1(n)$.
As $f_2$ is doubly-exponential, $h_1$ is exponential and $f_1$ is triply-exponential we get that $h_3$ is triply-exponential.
Also, the whole construction works in time at most $h_2(\size(G)) \cdot g(s)$ needed to construct all the thin 1-GVAS $H_X$
plus $s_2$. 

Let $h_4(n) := h_3(n) + h_2(n) = f_2(n + h_1(n))) + (n-1) \cdot f_1(n) + h_2(n)$.
We claim that Lemma~\ref{lem:top-branching} works for triply-exponential function $f := h_4$.
Clearly $h_3(\size(G)) \leq h_4(\size(G))$, so
\[
\size(H) \leq s_4 := h_4(\size(G)) + \size(G) \cdot \sum_{V \in \low(G)}\size(H_V)
\]
We also have that the whole construction works in time at most
\[
h_2(\size(G)) \cdot g(s) + s_2 \leq (s_2 + h_2(\size(G))) \cdot g(s) \leq s_4 \cdot g(s_4),
\]
as
\[
s_2 + h_2(\size(G)) \leq h_3(\size(G)) + h_2(\size(G)) +  \size(G) \cdot \sum_{V \in \low(G)}\size(H_V) \leq s_4
\]
and
\[
s \leq f_1(\size(G)) + \size(G) \cdot \sum_{V \in \low(G)}\size(H_V) \leq s_4.
\]
This finishes the proof of Lemma~\ref{lem:top-branching}.

\section{Proof of Lemma~\ref{lem:far-from-axis}}\label{sec:far-from-axis}
Let $G$ be a top-branching 1-GVAS $G$.
We prove Lemma~\ref{lem:far-from-axis} by showing that we can compute in exponential time an exponential $B \in \N$
such that $\reach{G} \cap [B, \infty)^2$ is semilinear and its representation can be computed in exponential time.

The following proposition shows that the above easily implies Lemma~\ref{lem:far-from-axis}.

\begin{proposition}\label{prop:semilin-to-thin}
For each semilinear, diagonal relation $R \subseteq \N^2$ one can construct a thin 1-GVAS $G$ of size polynomial
wrt. the size of the representation of $R$, such that $\reach{G} = R$.
\end{proposition}

\begin{proof}
It is enough to show a thin 1-GVAS for each linear, diagonal relation $R \subseteq \N^2$, as a finite union of thin 1-GVAS
is also a thin 1-GVAS. Let $R = b + P^*$, where $b = (\ell, r)$, $P = \set{p_1, \ldots, p_k}$ and for each $i \in [1,k]$
we have $p_i = (\ell_i, r_i)$. Notice that we can for free assume that one of the periods equals $(1,1)$, as the relation $R$
needs to be diagonal. We define the following 1-GVAS $G$. It has two nonterminals $X$ and $Y$, the initial one is $X$.
The rules are:
\[
X \to -\ell \, Y \, r \hskip 2cm Y \to 0 \hskip 2cm Y \to -\ell_i \, Y \, r_i \hskip 0.5cm \text{ for each } i \in [1,k].
\]
It is easy to observe that indeed $\reach{G} = R$.
\end{proof}

It is therefore enough to compute the mentioned semilinear representation of $\reach{G} \cap [B, \infty)^2$.
To show it, we will use Lemmas~\ref{lem:triangle}~and~\ref{lem:diagonal-lines}.
First, we define an upper-triangle.
For $a,\delta \in \N$, we call \emph{upper-triangle} a set of the form $\ut(a, \delta) := \{(x,y) \in \N^2 \mid x \ge a, y \ge x+ \delta\}$.
It is the upwards infinite triangle on the right of the vertical line $x = a$ and above the diagonal line $y = x + \delta$. 

The following lemma takes care of an upper triangle.

\begin{lemma}\label{lem:triangle}
For each top-branching 1-GVAS with reachability relation $R$ one can compute in exponential time some
at most exponential $a, \delta \in \N$
such that $R \, \cap \, \ut(a, \delta)$ is diagonal and semilinear and its representation is computable in exponential time.
\end{lemma}

We discuss the proof of Lemma~\ref{lem:triangle} below in Section~\ref{sec:triangle}.

Notice now that, similarly as in Section~\ref{sec:overview} we can apply the same technique
to $G_\rev$ and show that $\reach{G} \cap \mbox{LT}(a, \delta)$ is also diagonal, semilinear and its representation is computable
in exponential time, where $\mbox{LT}(a, \delta) = \{(x,y) \in \N^2 \mid y \geq a, x \geq y + \delta\}$.

Once we have covered $\mbox{UT}(a, \delta)$ and $\mbox{LT}(a, \delta)$ to cover the whole area $[B, \infty)^2$
for some $B$ we need to cover diagonal lines. The following, easy to show, lemma takes care of all the diagonal lines from some point on.

\begin{lemma}\label{lem:diagonal-lines}
For each 1-GVAS $G$ with reachability relation $R$ and each $\Delta \in \Z$ such that
$R(x, x + \Delta)$ for some $x \in \N$, there is an exponential $a \in \N$, computable in exponential time, such that $R(a, a + \Delta)$.
\end{lemma}

\begin{proof}
We aim at showing that for every $\Delta \in \Z$, there is always a derivation $\tau$ of size at most exponential in $\Delta$
and $\size(G)$ such that the effect of $\tau$ equals $\Delta$. Let then $a$ be the sum of the absolute values
of all the nonterminals in $\tau$, it is as well at most exponential. Then $\tau$ is valid at $a$, so $R(a, a+\Delta)$, as needed.

Therefore it remains to show that for each $\Delta \in \Z$ there is a small (meaning here exponential in $|\Delta|$ and $\size(G)$)
derivation of the effect $\Delta$.
Our proof is inspired by the proof of Parikh theorem about semilinearity of the Parikh image of a context-free grammar~\cite{DBLP:journals/jacm/Parikh66}.
Let a derivation $\tau$ be \emph{irreducible} if removing any simple cycle from $\tau$ decreases the set of nonterminals
in it. Let $G$ contain $n$ nonterminals. Then any path in any irreducible derivation has length at most $n^2$.
Otherwise, if there is a path of length at least $n^2+1$, then there is a nonterminal $X$ on that path occurring at least $n+1$ times.
Therefore, some of at least $n$ $X$-cycles in between these $X$-labelled nodes can be removed without decreasing the set
of nonterminals occurring in the derivation. Thus, any irreducible derivation is of at most exponential size, as summarised in this claim (used also later in Section~\ref{subsec:lemma_from_supertree}).

\begin{claim}\label{cl:irreducible-size}
Every irreducible derivation is of at most exponential size with respect to the number of nonterminals in the grammar.
\end{claim}
 
Take now a derivation $\tau$ of minimal size such that the effect of $\tau$ equals $\Delta$.
It can be decomposed into an irreducible derivation $\tau'$ with the same set of nonterminals as $\tau$
and a set of simple cycles $\sigma_1, \ldots, \sigma_m$.
Note that for any subset of these simple cycles one can obtain a derivation $\tau''$ obtained from $\tau'$ by
pasting into it the chosen cycles. This is because the set of nonterminals in $\tau$ and in $\tau'$ is the same,
so every cycle removed from $\tau$ can be pasted into $\tau'$.
For each $i \in [1,m]$, the effect of the simple cycle $\sigma_i$ is at most exponential; let $M$ be the maximal absolute value of such an effect.
Observe that the effect of each $\sigma_i$ is non-zero, as otherwise one can paste into $\tau'$ all the cycles
without this zero-effect cycle and obtain a derivation of effect $\Delta$, but of size smaller than $\tau$.
Notice also, that there cannot be two numbers $\Delta_1 < 0$ and $\Delta_2 > 0$ such that there are at least $M$ cycles
of the effect $\Delta_1$ and at least $M$ cycles of the effect $\Delta_2$. Indeed, if this would be the case, then
we remove $|\Delta_2|$ cycles of the effect $\Delta_1$ and $|\Delta_1|$ cycles of the effect $\Delta_2$
and do not change the total effect. Therefore, either for each $i \in [1,M]$ there is less than $M$ cycles of that effect,
or for each $i \in [-M,-1]$ there is less than $M$ cycles of that effect. Assume wlog. the first option. Thus there are at most $M^2$
positive cycles, so the total effect of positive cycles is at most $M^3$. This means that the number of negative cycles should also
be bounded, as each of them decreases the total effect by at least one. Indeed, let $K$ be the maximal absolute value of the effect
of any irreducible derivation. Then, the effect of the irreducible derivation plus all the positive cycles is at most $K + M^3$.
Recall that the total effect of the irreducible derivation plus all the cycles is $\Delta$.
Thus, the number of negative cycles is at most $M^3 + K + |\Delta|$, which is also exponential.
Therefore, the derivation $\tau$ consists the irreducible derivation $\tau'$, of size at most exponential, plus at most
$M^3 + K + |\Delta| + M^2$ simple cycles, each of size at most exponential. Thus the size of $\tau$ is exponential, which finishes
the proof.
\end{proof}

For each $\Delta \in \Z$, let $a_\Delta$ be the constant delivered by Lemma~\ref{lem:diagonal-lines}.
Let Lemma~\ref{lem:triangle} deliver a semilinear representation for $\reach{G} \, \cap \, \ut(a, \delta)$ and its
symmetric version deliver a semilinear representation for $\reach{G} \, \cap \, \mbox{LT}(a', \delta')$.
We set $B$ to be the maximum of $a$, $a'$ and $a_\Delta + |\Delta|$ for all $\Delta \in [-\delta', \delta]$.
It is easy to see that the upper-triangle $\ut(a, \delta)$, the lower-triangle $\mbox{LT}(a', \delta')$
and the lines of the form $(x, x+\Delta)$ for $x \geq a_\Delta$ cover the whole area $[B, \infty)^2$.
Moreover, the semilinear representation of this sum in indeed exponential, which implies
that the semilinear representation of it intersected with $[B, \infty)^2$ is also exponential.
Therefore, application of Proposition~\ref{prop:semilin-to-thin} finishes the proof of Lemma~\ref{lem:far-from-axis}.

\subsection{Proof of Lemma~\ref{lem:triangle}}\label{sec:triangle}
Let $G$ be a 1-GVAS and $R$ its reachability relation.
Let us fix $d \in \N$ to be the greatest common divisor of all the effects of the simple cycles of $G$. Notice that $d$ divides the effect of every cycle of $G$, as every cycle can be decomposed into simple cycles. Also, for every top nonterminal $X$, $G_X$ has the same cycles as $G$, so in particular also the same gcd.
The following claim makes an important observation about top-branching 1-GVAS.

\begin{claim}\label{cl:same-residuum}
    In a top-branching 1-GVAS $G$, all the derivations have the same effect modulo $d$.
\end{claim}

\begin{proof}
    Let $S$ be the initial nonterminal of $G$. Since the top component is branching, there is a derivation of the form $S \Rightarrow u_1 S u_2 S u_3$ for some $u_1,u_2,u_3 \in \Z^*$. 
    Let $S \Rightarrow v_1$ and $S \Rightarrow v_2$ be two complete derivations. 
    If we insert them into the previous derivation, we get the derivation $S \Rightarrow u_1 v_1 u_2 v_2 u_3$, which can be decomposed both as the derivation $S \Rightarrow v_1$ inserted into the cycle $S \Rightarrow u_1 S u_2 v_2 u_3$ and as the derivation $S \Rightarrow v_2$ inserted into the cycle $S \Rightarrow u_1 v_1 u_2 S u_3$. 
    Therefore $\sum v_1$ and $\sum v_2$ have same residuum modulo $d$.
\end{proof}

We denote by $r$ the residuum modulo $d$ of all derivations of $G$, and by $r_X$ that of $G_X$. Note that, contrary to $d_X$ which is equal to $d$ when $X$ is a top nonterminal, $r$ and $r_X$ are different in general: in the 1-GVAS of rules $X \to XY1, Y \to X1$ and $X \to 0$, we have $d =2$, and $X$-derivations have even effects while $Y$-derivations have odd effects.

The following characterisation if very useful in the rest of our paper.
In the sequel, call a 1-GVAS \emph{infinitary} if it satisfies the five equivalent properties
listed in Proposition~\ref{prop:branching-top-component}.

\begin{proposition} \label{prop:branching-top-component}
    The following are equivalent for a 1-GVAS $G$:
    \begin{enumerate}
        \item The supremum of effects of cycles in $G$ is infinite.
        \item There exists a simple cycle with positive effect.
        \item There exists a simple $S$-cycle with a positive global effect and a positive left effect, which has size
        at most exponential wrt. the size of $G$.
        \item There exists a simple $S$-cycle with a positive global effect and a positive left effect.
        \item There exists $a \in \N$ such that $\vert R(a) \vert = \infty$
    \end{enumerate}
\end{proposition}

\begin{proof}
    The implication $5 \Rightarrow 1$ is immediate.
    
    The implication $1 \Rightarrow 2$ is easily shown by contraposition: if no simple cycle has positive effect, then the derivations of maximal effects are the simple derivations, and there are only finitely many of them.

    \textbf{$2 \Rightarrow 3$:} By hypothesis there is a derivation $D_1: V \Rightarrow v_1 V v_2$ with $\sum v_1 + \sum v_2 > 0$. Take also derivations of the form
    \begin{itemize}
        \item $D_2: S \Rightarrow u_1 S u_2 S u_3$ with $u_1,u_2,u_3 \in \Z^*$
        \item $D_3: S \Rightarrow w_1 V w_2$ with $w_1, w_2 \in \Z^*$
        \item $D_4: V \Rightarrow w_3$ with $w_3 \in \Z^*$
    \end{itemize}
    One can easily that all the derivations $D_1$, $D_2$, $D_3$ and $D_4$ can be assumed to be of at most exponential size.
    Indeed, for derivations $D_2$, $D_3$ and $D_4$ it is immediate. For $D_1$ is suffices to take any derivation
    $D: V \Rightarrow v_1 V v_2$ with $\sum v_1 + \sum v_2 > 0$ and remove simple cycles with nonpositive effect
    as long as they exist. Then remove simple cycles with positive effect until the total effect is positive. As each simple cycle
    is at most exponential then the effect at the end is at most exponential and the derivation $D_1$ consist therefore of a simple
    derivation and at most exponentially many simple cycles (as each of them increases the total effect).

    Now, we insert $D_4$ into $D_3$ and in between of them $n$ copies of $D_1$. Inserting this block on the left into $D_2$
    yields a derivation of the form $S \Rightarrow u_1 v_1^n w_1 w_3 w_2 v_2^n u_2 S u_3$. For $n$ big enough both
    effect of $u_1 v_1^n w_1 w_3 w_2 v_2^n u_2$ and of $u_1 v_1^n w_1 w_3 w_2 v_2^n u_2 u_3$ is positive. It is
    easy to see that an exponentially big $n$ suffices, which finishes the proof of the implication.
    
    \textbf{$3 \Rightarrow 4$:} This implication is immediate.

    \textbf{$4 \Rightarrow 5$:} Take $S \Rightarrow u S v$ a cycle with $\sum u > 0$ and $\sum u + \sum v > 0$, $S \Rightarrow w$ any complete derivation and $a \in \N$ large enough so that the run $u$ is valid from $a$. There is $n_0 \in \N$ such that all the derivations $S \Rightarrow u^n w v^n$ with $n \ge n_0$ have different effects and yield runs valid from $n$. Hence, $\vert R(a) \vert = \infty$.
\end{proof}


With Proposition~\ref{prop:branching-top-component} in hand we can continue the proof of Lemma \ref{lem:triangle}.
Recall that $G$ is a 1-GVAS and $R$ is its reachability relation.
We can compute all the simple cycles to determine whether $G$ is infinitary.
If it is not, then the effects of derivations generated by $G$ are upperbounded, and the maximal ones are reached by simple derivations. 
Therefore in that case it is enough to compute all the simple derivations of $G$ to find $\delta \in \N$ such that $\mbox{UT}(0, \delta) = \emptyset$.
\newline

Now, suppose that $G$ is infinitary. 
We want to find an upper triangle where the reachability relation is a computable semilinear set $L$.
Recall that, by definition of $d$ and $r$, $R \subseteq \{(x,y) \in \N^2 \mid y \in x + r + d \Z\}$. We will choose an upper triangle where this overapproximation is actually reached. This will prove the Lemma~\ref{lem:triangle}.

In fact, it even suffices to find a threshold $T \in \N$ such that the overapproximation is reached on the vertical line $x=a$ above threshold $T$, namely $R(a) \, \cap \, [T, \infty) = (a + r + d\Z) \, \cap \, [T, \infty)$. Indeed, recall that reachability relations are diagonal, that is, stable by addition of $(1,1)$. Therefore, if the overapproximation is reached on the vertical line $x=a$ above threshold $T$, then it is reached on the upper triangle $\text{UT}(a,T-a)$.

We show it by the use of Lemma~\ref{lem:line-linear-set}.
We formulate Lemma~\ref{lem:line-linear-set} separately as it will be also used in Section~\ref{subsec:lemma_from_supertree}
in a similar reasoning.
More concretely, to show Lemma \ref{lem:triangle} we use Lemma \ref{lem:line-linear-set} in the following way.
We apply it to $G$ and $a$, with $X$ equal to starting nonterminal of $G$,
with the cycle $\gamma$ from condition 3 in Proposition~\ref{prop:branching-top-component},
and with $\tau$ being any simple $X$-tree.

\begin{lemma} \label{lem:line-linear-set}
    Let $G$ be a top-branching 1-GVAS and $X$ be a top nonterminal.
    Let $\gamma$ be an $X$-cycle of a positive left effect, and positive global effect $p$ and let $a \in \N$
    be such that there is a complete derivation $\tau$ valid at $a$
    and containing the cycle $\gamma$ (say, between nodes $N$ and $M$).
    Then, there is a threshold $T \in \N$ such that $R(a)  \, \cap \, [T, \infty) = (a + r + d\Z)  \, \cap \, [T, \infty)$.
    Moreover, $T$ is of size exponential in $\size(G)$ and polynomial in $p,  s_{\text{left}}$ and $s_{\text{right}}$,
    where $s_{\text{left}}$ is the sum of the terminals on the left of the subtree of $\tau$ rooted at $N$, and
    $s_{\text{right}}$ the sum of the absolute values of the terminals on its right.
\end{lemma}

\begin{proof}
Recall that the effect of every $X$-derivation has residuum $r$ modulo $d$. Therefore for any number $T \in \N$
we have $R(a)  \, \cap \, [T, \infty) \subseteq (a + r + d\Z)  \, \cap \, [T, \infty)$. Our aim is therefore to find $T$ big enough such that
the converse inclusion also holds.

    First, observe that the run yielded by $\tau$ remains valid from input $a$ if we insert more copies of the cycle $\gamma$ at node $M$ (or node $N$). It also remains valid if we replace the tree rooted at $M$ by another $X$-tree and insert enough copies of $\gamma$.

    By replacing the tree rooted at $M$ by some simple $X$-tree, and inserting arbitrarily many copies of the cycle $\gamma$ above, we already generate derivationss witnessing that $R(a) \cap [T, \infty) \supseteq (a + r + p\Z) \cap [T, \infty)$ for some $T$.
    We would like to have the same inclusion, but with the smaller period $d$ instead of $p$.
    For that, it suffices to find a way to increase the effect of our derivations by $d$, or simply by some sufficiently small value in $d + p \Z$.

    Here is how we proceed. 
    Let $k \in \N$ be the integer such that $p = dk$. We need to create for each $i \in [0, k-1]$ a derivation $\tau_i$
    which is valid from the input $a$ and has effect in $i \cdot d + p \Z$.    
    Let $c_1,\ldots,c_m$ be the effects of the simple cycles in $G$. By Bezout's theorem, there are $k_1, \ldots, k_m \in \Z$ such that $d = \sum_{j=1,\ldots,m} k_j c_j$. By adding a multiple of $p$ large enough to each $k_j$, we find $k_1', \ldots, k_m' \in \N$ such that $d \equiv \sum_{j=1,\ldots,m} k_j' c_j \mod p$. This intuitively means that combining $k'_j$ copies of cycles of effects $c_j$ gives us the
    effect in $d + p \Z$. 
    All those numbers can be computed in time exponential in $\size(G)$ 
    using the extended Euclidean algorithm.
    Let $\sigma$ be an $X$-tree where every nonterminal occurs at least once.
    There are such $\sigma$ of size exponential in $\size(G)$.
    For $i \in [0,k-1]$ let $\sigma_i$ be a copy of $\sigma$ where we insert, for every $j \in [1,m]$, exactly $i \cdot k_j'$
    copies of a cycle of the effect $c_j$.
    Then the effect of $\sigma_i$ is in $i \cdot d + p\Z$. Using $\sigma_i$ we can pretty easily get a required threshold $T$.
    We create the needed derivation $\tau_i$ as a copy of $\tau$, where we replace the tree rooted at $M$ by $\sigma_i$, and then insert above the least number of copies of $\gamma$ necessary to make the whole derivation valid from input $a$.
    Let us take $T$ to be the maximal effect of a derivation $\tau_i$, for $i \in [0,k-1]$.
    Then indeed we have $R(a) \cap [T, \infty) = (a + r + d\Z) \cap [T, \infty)$, as required.
    

    To bound the size of $T$, there only remains to bound the number of times that the cycle $\gamma$ has to be inserted in the $\tau_i$.
    This number is at most $s_{\text{right}}$, plus some exponential function of $\size(G)$. Hence, $T$ has the desired size.



\end{proof}

\section{Proof of Lemma \ref{lem:small-lines}}\label{sec:mainproof}

In this section, we prove Lemma \ref{lem:small-lines} in the case where $G$ is infinitary, namely it satisfies the property of Proposition \ref{prop:branching-top-component}. Otherwise, the proof is trivial: there is a bound $\Delta \in \N$ on the effects of the derivations produced by $G$, so for all $a \in \N$, $R_G(a)\vert_{\ge a + \Delta} = \emptyset$. This bound is reached at simple derivations, so it is easily computable.



\subsection{General idea} \label{subsec:general-idea}

Recall that we have some input value $a \in \N$, a top-branching 1-GVAS $G$ and for every lower nonterminal $V$ a thin 1-GVAS $H_V$ equivalent to $G_V$. Our goal is to construct a thin 1-GVAS $H$,
which is almost-equivalent to $G$ when the input is exactly $a$. More concretely, it can output the same
big values as $G$ and some small values as $G$, but should output nothing more. It should also output nothing more with other input values.

\paragraph*{Searching through derivations}
Our general idea to construct $H$ is to inspect derivations of $G$ and simplify them.
For some very special $G$ this might work directly, even without simplifications.
If, for example, there are only finitely many derivations of $G$ and we are able to compute them all,
then $H$ can be simply constructed as follows:
it has an initial nonterminal $S$ and for each output $b$ produced by $G$ it has a rule $S \to (-a)ab$. (We need the $-aa$ because the effect $b$ might not be achievable with an input smaller than $a$, and $H$ must output nothing more than $G$ on every input.)
Obviously such a grammar is thin, but also obviously this is a very naive approach.

In general $G$ may have infinitely many derivations. Thus, grammar $H$ cannot just simply list them all.
Here we need some observations how we can search through derivations and during the construction
of a derivation realise that we can produce an almost-equivalent thin 1-GVAS (by almost-equivalent
be mean $H$ such that $H \under_{\set{a} \times \N_{\geq T}} G$ for some computable threshold $T$, where
$\N_{\geq T} = \set{n \in \N \mid n \geq T}$). Intuitively, we will construct derivations in the left-to-right manner.
To make this precise we define an \emph{Euler Tour} of a derivation. Recall that each derivation of $G$ is a full binary tree.

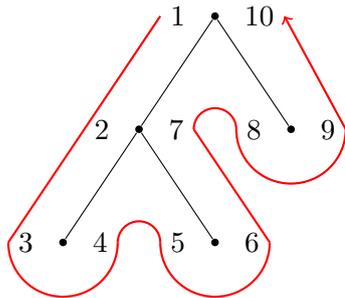
\begin{wrapfigure}{l}{0.4 \textwidth}
\centering
\begin{tikzpicture}[
    sibling distance=20mm, level distance=15mm]

\node[circle,fill, minimum size=1mm, inner sep=0pt] (root) {} 
    child {
        node[circle,fill, minimum size=1mm, inner sep=0pt] (left) {} 
            child { node[circle,fill, minimum size=1mm, inner sep=0pt] (leftleft) {} } 
            child { node[circle,fill, minimum size=1mm, inner sep=0pt] (leftright) {} } 
    }
    child {
        node[circle,fill, minimum size=1mm, inner sep=0pt] (right) {} 
    };

\node[left=2mm of root] (1) {1};
\node[right=2mm of root] (10) {10};

\node[left=2mm of left] (2) {2};
\node[right=2mm of left] (7) {7};

\node[left=2mm of right] (8) {8};
\node[right=2mm of right] (9) {9};

\node[left=2mm of leftleft] (3) {3};
\node[right=2mm of leftleft] (4) {4};

\node[left=2mm of leftright] (5) {5};
\node[right=2mm of leftright] (6) {6};

\draw[red, thick] (1.west) to (3.west);

\draw[red, thick] (3.west) arc[start angle=180, end angle=360, radius=0.72cm];
\draw[red, thick] (5.west) arc[start angle=180, end angle=360, radius=0.72cm];
\draw[red, thick] (8.west) arc[start angle=180, end angle=360, radius=0.72cm];

\draw[red, thick] (4.east) arc[start angle=180, end angle=0, radius=0.28cm];
\draw[red, thick] (7.east) arc[start angle=180, end angle=0, radius=0.28cm];

\draw[red, thick] (6.east) to (7.east);

\draw[red, thick, ->] (9.east) to (10.east);

\end{tikzpicture}
\end{wrapfigure}

The \emph{Euler Tour} of a full binary tree $\tau$, written $\mbox{ET}(\tau)$, is a sequence of actions, where each action is either the first visit or the last visit of a node. For a leaf $L$, we have $\mbox{ET}(L) := \text{first}(L), \text{last}(L)$ and for a tree $\tau$ of root $N$, of left subtree $\tau_1$ and right subtee $\tau_2$ we have $\mbox{ET}(\tau) := \text{first}(N) \cdot \mbox{ET}(\tau_1) \cdot \mbox{ET}(\tau_2) \cdot \text{last}(N)$. In the picture above for the five-node derivation the sequence of actions is indicated by the numbers from $1$ to $10$.

We will be constructing derivations following the order of the Euler Tour. At each moment of such a construction the part to the left
of the currently processed node will be fully constructed, while the part to the right will be only partially constructed. More precisely,
we call the currently processed node the \emph{current node} and the path from the root to the current node the \emph{current branch}.
The part to the right of the current branch consists only of children of the nonterminals on the current branch.
We call such a partially constructed derivation a \emph{partial derivation}.
Example of a partial derivation is on Figure~\ref{fig:partial-deriv}, the input $a$ is assumed to be $10$.

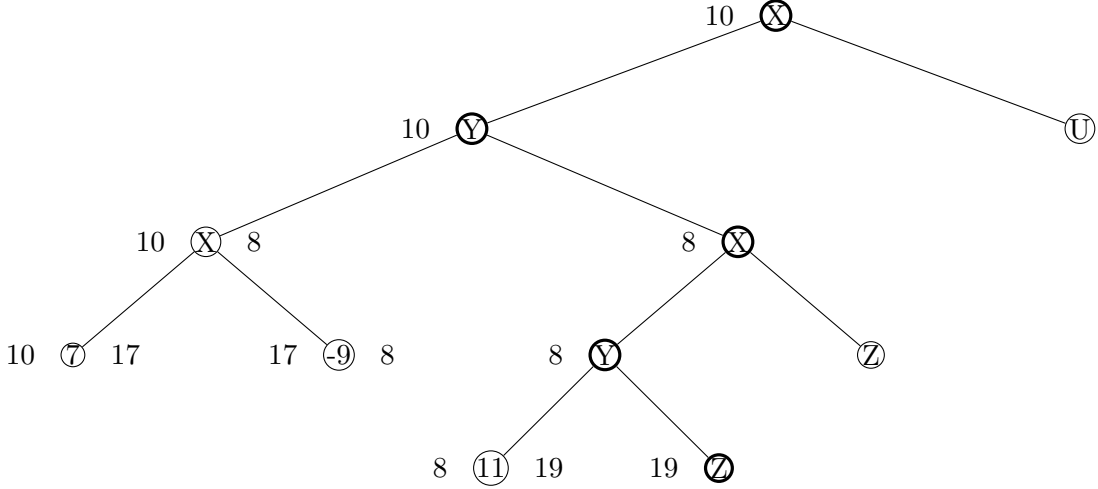
\begin{figure}[h]
\caption{Partial derivation}
\label{fig:partial-deriv}
\begin{tikzpicture}[
    level distance=15mm,
    level 1/.style={sibling distance=8cm},
    level 2/.style={sibling distance=7cm}, 
    level 3/.style={sibling distance=3.5cm}, 
    level 4/.style={sibling distance=3cm}]

\node[circle, very thick, draw, minimum size=1mm, inner sep=0pt] (root) {X} 
    child {
        node[circle, very thick, draw, minimum size=1mm, inner sep=0pt] (L) {Y}
            child { node[circle, draw, minimum size=1mm, inner sep=0pt] (LL) {X} 
            	child { node[circle, draw, minimum size=1mm, inner sep=0pt] (LLL) {7} } 
            	child { node[circle, draw, minimum size=1mm, inner sep=0pt] (LLR) {-9} }
            }
            child { node[circle, very thick, draw, minimum size=1mm, inner sep=0pt] (LR) {X}
            	child { node[circle, very thick, draw, minimum size=1mm, inner sep=0pt] (LRL) {Y}
 	           	child { node[circle, draw, minimum size=1mm, inner sep=0pt] (LRLL) {11} } 
      		      	child { node[circle, very thick, draw, minimum size=1mm, inner sep=0pt] (LRLR) {Z} }      
		} 
            	child { node[circle, draw, minimum size=1mm, inner sep=0pt] (LRR) {Z} }            
            }
    }
    child {
        node[circle, draw, minimum size=1mm, inner sep=0pt] (R) {U}
    };

\node[left=2mm of root] (1) {10};
\node[right=2mm of root] (10) {};

\node[left=2mm of L] (2) {10};
\node[right=2mm of L] (7) {};

\node[left=2mm of R] (8) {};
\node[right=2mm of R] (9) {};

\node[left=2mm of LL] (3) {10};
\node[right=2mm of LL] (4) {8};

\node[left=2mm of LLL] (3) {10};
\node[right=2mm of LLL] (4) {17};

\node[left=2mm of LLR] (3) {17};
\node[right=2mm of LLR] (4) {8};

\node[left=2mm of LR] (5) {8};
\node[right=2mm of LR] (6) {};

\node[left=2mm of LRL] (5) {8};
\node[right=2mm of LRL] (6) {};

\node[left=2mm of LRLL] (5) {8};
\node[right=2mm of LRLL] (6) {19};

\node[left=2mm of LRLR] (5) {19};
\node[right=2mm of LRLR] (6) {};

\end{tikzpicture}
\end{figure}

In the example on Figure~\ref{fig:partial-deriv} the current node is the $Z$-labelled node at the lowest level.
The nodes at the current branch are
emphasised by a thick circle. Notice that nodes on the current branch have only their input defined, nodes on the left of the current
branch have both input and output defined, while nodes on the right of the current branch have neither input nor output defined.

To guarantee termination of our algorithm we need to find situations in which we can stop building a derivation and claim that
all the derivations, which can be constructed from the current derivation can be expressed by some rule of some thin 1-GVAS.
We present now two such situations. In fact, in the algorithm presented below, they need to be modified a bit, but the intuition
behind them will still be the same. Therefore, we simplify the conditions here and later comment how they should be modified
to accordingly deal with some special situations.

\paragraph*{Two stop conditions}
The first situation in which we can stop expanding a partial derivation is when the current nonterminal and all the nonterminals on the right
side of the current branch are lower nonterminals. For example, if both $U$ and $Z$ on Figure~\ref{fig:partial-deriv} would be lower nonterminals then
we stop expanding this partial derivation. This is because if in the constructed grammar we add a rule $X \to (-10)(10) 9 Z Z U$
it cares of all the derivations, which can be constructed from our current derivation, and also, more importantly, it contains only lower nonterminals
on the right hand side. Therefore a grammar constructed from such rules is always thin.

The second situation in which we can stop expanding a partial derivation is when the current nonterminal is a top nonterminal
and additionally the current input value is high enough. For example, if $Z$ on Figure~\ref{fig:partial-deriv} would be a top nonterminal and
value $19$ high enough with respect to $G$, then we stop expanding this partial derivation.
The intuition is that if the current nonterminal, say $Z$, is a top nonterminal, then we can
use condition 3 in Proposition~\ref{prop:branching-top-component} to guarantee that there is a $Z$-cycle $\sigma$
with a positive global effect and a positive left effect. This would allow us to pump up the counter value.
However, if the input value of the current node is not high enough then $\sigma$ may not be fireable in the current situation.
Fortunately, because of condition 3 in Proposition~\ref{prop:branching-top-component}
there is an exponential bound on the size of $\sigma$, depending only on the size of 1-GVAS.
If the current input value is at least as big as the maximal drop of the counter, which can be 
experienced along the left part of $\sigma$, then we can safely fire $\sigma$. Moreover, then we can fire $\sigma$ arbitrarily many times,
so indeed pump the input counter arbitrarily high. It is not immediate to see, but in such a situation we are able to directly
construct a thin 1-GVAS $H$ almost-equivalent to $G$ when starting at input $a$.

The idea is that if we have such a pumping-up fireable cycle then the set of all the possible output values of all the derivations
with input $a$ is semilinear and computable. Then we use Proposition~\ref{prop:semilin-to-thin} to construct a thin 1-GVAS
with the relation $\set{a} \times S$, where $S$ is the set of all the possible output values. In order to see why having a pumping-up
cycle we use Lemma~\ref{lem:line-linear-set}. Intuitively, it says that having a pumping-up cycle is enough to get all the possible
effects of derivations from some not very high threshold (an exponential one).
To summarise the above considerations, the following claim shows that we may stop expanding the current partial derivation
as soon as we encounter a top nonterminal with big input value.

\begin{claim}\label{cl:high-input-branching}
There is a bound $A \in \N$, exponential in the size of $G$, such that
for any partial derivation $\tau$ with current node labelled by a top nonterminal and current input value
at least $A$ the set of all the output values of derivations of $G$ equals the set of all the output values
obtained from $\tau$ by pumping-up a single cycle of at most exponential size and pasting some small cycles of at most
exponential size. Moreover, all the cycles are computable in exponential time,
so in consequence having a top nonterminal and input value at least $A$ one can compute the set of $R_G(a) \cap [T, \infty)$
for some exponential $T$ in exponential time.
\end{claim}

Notice now that the two described situations are pretty different in style. In the first situation we may stop expanding the current
partial derivation, but we need to expand others to guarantee that all the possible derivations have been considered. On the other hand,
encountering the second situation guarantees that we may stop searching for any other derivations, as we can directly compute
a thin 1-GVAS $H$ almost-equivalent to $G$.

\paragraph*{Bounding the current node depth}
The two above stop conditions are however not sufficient for termination of the algorithm. One may expand a partial derivation
further and further and never meet any of those conditions. Therefore, to guarantee termination, we design a technique,
which allows us to consider only partial derivations with the current node being on a rather small depth.
To understand the need for this technique, consider a rule of the form $X \to X 1$. If we apply this rule in our partial derivation
not much changes. We have a new node, but it is labelled by the same nonterminal and has the same input as its parent.
So, if we don't deal with such cases appropriately, we can expand our partial derivation forever. The idea to avoid this problem
is to not allow for cycles with the left effect equal zero in the partial derivation.
We call cycles with the left effect equal zero the left-zero-effect cycles.
Of course, we cannot just disallow such cycles in the derivations, because sometimes they are needed.
However, we can disallow them in the partial derivation.
Intuitively, instead of adding a left-zero-effect $X$-cycle when we visit an $X$-labelled node
for the first time in the Euler Tour, we add the same $X$-cycle while visiting the same node for the last time in the Euler Tour.
That simplifies our partial derivations much, but adds to the complication of the algorithm, as while going upwards in the Euler Tour we
need to add some cycles. At the moment however, we focus on the simplification.

In the process of expanding the partial derivations we treat nodes labelled by lower nonterminals a bit like nodes labelled by terminals. Indeed, for every $V \in \low(G)$, with have a thin 1-GVAS $H_V$ equivalent to $G_V$, and by~\cite{DBLP:conf/fsttcs/AtigG11} we have a reachability algorithm for them. Therefore, we focus on the situation when all the nonterminals labelling the ancestors of the current node are top nonterminals; this is indeed the case in which we deal in the algorithm.
Notice now that by Claim~\ref{cl:high-input-branching} if any of the input values on the current branch is at least $A$
then we can finish our algorithm. If this is not the case all the nodes on the current branch, above the current node,
have input values in the interval $[0, A-1]$.
Notice, that there can be no two nodes on the current branch labelled by the same nonterminal and with the same input value,
as that would mean that we have created a cycle in between them with the left effect being zero. Thus, the number of the nodes on the
current branch is bounded by $A \cdot \size(G)$. In other words, depth of the current node is always at most $A \cdot \size(G)$.

\paragraph*{Searching through the partial derivations}
Our algorithm searches through the partial derivations. It organises these partial derivations in a tree, which we call a \emph{supertree},
to distinguish it from the derivation, which are also trees. Nodes of the supertree are called \emph{supernodes},
branches of the supertree are called \emph{superbranches} etc. In each supernode of the supertree there is a partial derivation.
For each possible expansion of this partial derivation we create a child of its supernode and place there the expanded derivation.
Expansions can be, roughly speaking, of two types: either we proceed downwards or upwards the partial derivation.
We proceed downwards if the current node in the partial derivation has specified input, but unspecified output.
If the current node is $X$-labelled we then look for all the rules $X \to \alpha$ and expand each independently, so for each one
create a new superchild. If the newly created node in such a superchild has the same nonterminal and input value as some of its
ancestors then we call this supernode a \emph{failure}, as it contains a left-zero-effect cycle. Otherwise we treat it as mentioned till now.
On the other hand, if both the input and output of the current node are specified we proceed upwards.
This is now a bit challenging, as, recall, we need to paste the left-zero-effect $X$-cycles on our tour upwards.

Seemingly the step of going upwards can be problematic, as we can add as many left-zero-effect cycles, as we want.
However, it turns out that if the output value is big enough then we can always finish our algorithm, a bit similarly
as described in the Claim~\ref{cl:high-input-branching}. It is important to emphasise that we can do this even
if we are currently going upwards and the current node is not labelled by a top nonterminal.
The idea is that either from the current partial derivation we can reach later a top nonterminal or we cannot.
Importantly, we should take into account the left-zero-effect cycles, which while added can possibly insert some top nonterminal.
We take into account at every such supernodes all the left-zero-effect cycles, which are currently available, namely have
been encountered in any supernode of our supertree, which was created up to that point (recall that such supernodes are called failures).
It is not hard to check whether a top nonterminal is reachable from the current partial derivation. 
We just need to see what is the nonterminal is the current node, what are the nonterminals hanging on the right
of the current branch and what are the nonterminals on the main branch.
For example in Figure~\ref{fig:partial-deriv} the current nonterminal is $Z$, the nonterminals on the right are $Z$ and $U$
and the nonterminals on the current branch are $X$ and $Y$. If the current nonterminal or any of the nonterminals on the right is a top nonterminal,
then we can clearly reach a top nonterminal in the future. However, there is also another option. It might happen that $U$ and
$Z$ are lower nonterminals, but there is some left-zero-effect $X$-cycle, which produces on the right a top nonterminal.
Thus, we also need to check what can be reached from the nonterminals on the main branch by the use of currently available cycles.

There are two options. The first option is that only lower nonterminals can be reached. In that case we can stop exploring this
partial derivation and create appropriate thin grammar taking care of this situation (details will be explained in the construction).
The second option is that some top nonterminal can be reached. It is not hard to see that if such a nonterminal can be reached
using the nonterminals on the main branch (the number of them is bounded by $A \cdot \lvert \topp(G) \rvert$) and the simple cycles,
which we add then it can be reached rather quickly. More concretely, there is a computable bound, denoted below by $D'$,
such that if a top nonterminal can be reached then it can be reached in at most $D'$ moves. This means that in our current
partial derivation we are going upwards (so the output is specified) and the output is bigger than $A + C \cdot D'$, where $C$
is the maximal drop we can do in one move then we can reach the reachable top nonterminal with the counter value at least $A$.
This, in turn, means that are done, we can compute the whole set of reachable outputs by Claim~\ref{cl:high-input-branching}.

\paragraph*{Concluding the argument}
Having all these observations it is not hard to conclude the argument. Our algorithm builds a supertree, which is intuitively
just searching through all the possible partial derivations. It is not hard to show that the supertree is finite, as, intuitively
the size of the partial derivation grows along a superbranch and this size is bounded. Details are presented in Section~\ref{subsec:finite_supertree}.

In the constructed supertree we have finitely many supernodes and a partial derivation in each one.
If in some of these partial derivations we can reach a top nonterminal with input value at least $A$
then we use Claim~\ref{cl:high-input-branching} to directly compute the set of reachable outputs.
A supernode corresponding to such a situation is called in our algorithm a \emph{success} supernode.
Otherwise, if there is no success supernode in the whole supertree we consider all the supernodes in which the partial
derivation can reach only lower nonterminals. In the algorithm below, such supernodes are the neutral superleaves. (They are called neutral because their status is neither failed, nor successful.)
For each neutral superleaf we compute a thin 1-GVAS corresponding to all the derivations which can be created
from expanding the particular derivation of this supernode. Then the union of all of these thin 1-GVAS is the thin 1-GVAS
equivalent to the input 1-GVAS $G$.

\subsection{The supertree} \label{subsec:supertree}


In this subsection, we give a formal definition of the supertree. For that, we need two constants: 
\begin{itemize}
\item $A$ is a constant such that, for every top nonterminal $X$ of $G$, there is a cycle of positive left effect and positive global effect whose left part yields a run valid at $A$.
($A$ exists because we assume that $G$ is infinitary, namely satisfies the properties of Proposition \ref{prop:branching-top-component}.)
\item $C$ is a constant such that all terminals of $G$ are greater than $-C$ and all nonterminals can produce a complete derivation of effect at least $-C$ and valid at $C$.
\end{itemize}
We will also write $D := A \cdot \lvert \text{Top}(G) \rvert$ 
and $D' := D + (A \cdot \lvert \text{Top}(G) \rvert)^2$. It will represent the depth of the current nodes of partial derivations at different moments in the proof.
\newline

We need some discussion about the formalisation of trees. Binary trees are usually formalised as prefix-closed subsets of $\{0,1\}^*$,
but this definition is not well-suited for our purpose. Indeed, we manipulate derivation trees, some of which are expansions of others. For example, $S(X(X_1,X_2),Y)$ and $S(X(X_1',X_2'),Y)$ are two expansions of $S(X,Y)$. We would like to say that the nodes of label $X$ in those three derivation trees are ``the same node, but in different trees'', and yet ``not the same node'' as the node of label $X$ in $S(V,X)$ for instance.

The implicit formalisation of derivation trees in the remaining of this section is the following. We have an infinite set $\mathcal I$ of identifiers (natural numbers for instance). A binary tree is a triple $(I,E_\text{left},E_\text{right})$ where $I \subseteq \mathcal I$ is a finite set of identifiers and $E_\text{left},E_\text{right}$ are binary relations on $I$. We write $E := E_\text{left} \cup E_\text{right}$. We ask that $E_\text{left},E_\text{right}$ be the left-child and right-child relations of a full binary tree of nodes $I$, that is:
\begin{itemize}
\item There is a unique node $R$, called the root, such that for all $N \in I$, $(N,R) \notin E$.
\item For every other node $M$, there is a unique node $N$ such that $(N,M) \in E$. 
\item $E_\text{left} \cap E_\text{right} = \emptyset$
\item There is no $E$-cycle. That is, there is no sequence $N_1,...,N_k$ such that it holds that
$(N_1,N_2),\dots, (N_{k-1},N_k),(N_k,N_1) \in E$.
\end{itemize}
When $N \in I \cap I'$, then we will say that $N$ is a \emph{common node} of the trees $(I,E_\text{left},E_\text{right})$ and $(I',E_\text{left}',E_\text{right}')$. In the following, we will often say things like ``a copy of $(I,E_\text{left},E_\text{right})$, where $N$ has two children''. It will mean a tree of the form $(I \uplus \{N_1,N_2\},E_\text{left} \uplus \{(N,N_1)\},E_\text{right} \uplus \{(N,N_2)\})$ where $N_1,N_2 \in I$ are two identifiers that are still not used in any other trees we defined.
\newline

Now we can give the formal definition of the supertree. Each supernode has 
\begin{itemize}
\item a status, which can be \emph{neutral}, \emph{successful} or \emph{failed}
\item a derivation with input and output counters specified at some nodes, but not necessarily everywhere. For brevity, we call those derivations "partial derivations". We define them below.
\end{itemize}

A \emph{partial derivation} $\tau$ is a binary tree $(I,E_\text{left},E_\text{right})$ equipped with a labelling function 
$\text{label}: 
I \to \mathcal N(G) \cup \Z$ and partial input and output functions viewed as functions $\text{in, out}: 
I \to \N \cup \{\bot\}$.

We require flow conditions to be satisfied, namely
\begin{itemize}
    \item for every internal node $N$ with left child $N_L$ and right child $N_R$, we have $\text{in}(N_L) = \text{in}(N),\text{out}(N_R) = \text{out}(N)$, and $\text{in}(N_R) = \text{out}(N_L)$
    \item for every leaf $N$ labelled by a terminal $c$, $\text{out}(N) = \text{in}(N) + c$.
\end{itemize}
However, leaves are not necessarily all labelled by nonterminals.
Moreover, there must be a unique action $\mbox{Act}$ in the Euler Tour of $(I,E_\text{left},E_\text{right})$ such that for every node $N$:
\begin{itemize}
\item $\text{in}(N)$ is specified if and only if $\text{first}(N) \le \text{Act}$
\item $\text{out}(N)$ is specified if and only if $\text{last}(N) \le \text{Act}$
\item $N$ has children if and only if it is labelled by a top nonterminal and $\text{first}(N) < \text{Act}$.
\end{itemize}
The action $\mbox{Act}$ is called the \emph{current action} of the partial derivation, the associated node is called the \emph{current node}, and the branch from the root to the current node is called the \emph{current branch}. The three above properties
say that we construct partial derivations following the order of a Euler Tour. At the first visit, we define the inputs and expand the top nonterminals, and at the last visit we define the outputs. Note that only top nonterminals are expanded, so there are leaves labelled by lower nonterminals. A node of partial derivation is called \emph{unvisited} if neither its input nor its output is specified. For brevity, we sometimes refer to ``top/lower nodes'' to denote nodes labelled by top/lower nonterminals.

We say that a partial derivation $\tau$ is \emph{produced} by some GVAS $G$ if the derivation obtained by dropping the counters in $\tau$ is produced by $G$, and, for every leaf $N$ of $\tau$ labelled by a nonterminal $V$, $\text{in}(N) \xrightarrow{V} \text{out}(N)$.

For every $a' \in \N$ and $X \in \text{Top}(G)$, we call\emph{$(a',X)$-cycle} a partial derivation where both the current node and the root have input $a'$, label $X$ and unspecified output. Its current branch and current node are also called \emph{main branch} and \emph{distinguished leaf} to match the terminology of cycles. \emph{Partial cycles} are partial derivations that are $(a',X)$-cycle for some $a' \in \N$ and $X \in \text{Top}(G)$.
\emph {Simple partial cycles} are partial cycles where the root and the current node is the only pair of nodes in the main branch having same inputs and same labels. For $a' \in \N$ and $X \in \text{Top}(G)$, for all $(a',X)$-cycles $\gamma$ and $\gamma'$, we write $\gamma \equiv \gamma'$ if the main branches of $\gamma$ and $\gamma'$ induce the same sequence of inputs, labels, and labels of (potential) children on the right of the main branch. More precisely, if we write $N_0, ... N_k$ the sequence of nodes of the main branch of $\gamma$ from the distinguished leaf to the root, and $N_0', ... N_{k'}'$
the corresponding sequence in $\gamma'$, we ask that
\begin{itemize}
    \item $k' = k$
    \item for every $i \in [0,k], \text{in}(N_i), \text{label}(N_i) = \text{in}(N_i'), \text{label}(N_i')$
    \item for every $i \in [0,k-1]$, $N_i$ is a left child if and only if $N_i'$ is a left child; and in that case the right child of $N_{i+1}$ has the same label as the right child of $N_{i+1}'$.
\end{itemize}

So relation $\equiv$ is an equivalence relation. For fixed $a' \in \N$ and $X \in \topp(G)$, there are finitely many classes of simple $(a',X)$-cycles.

Given a partial derivation $\tau$ with current node $N$, whose parent is denoted $N'$, and a $(\text{in}(N),\text{label}(N))$-cycle $\gamma$ of root $M$, we call \emph{inserting $\gamma$ in $\tau$} the partial derivation obtained by
\begin{itemize}
    \item taking a copy of $\gamma$ where all nodes are renamed to unused identifiers, except the current node which is renamed to $N$
    \item taking the union of this with $\tau$ (union of the set of nodes and union of the child relations)
    \item removing the edge from $N'$ to $N$ and adding an edge of the same type from $N'$ to $M$.
\end{itemize}
Observe that $\gamma$ is inserted above the current node.
\newline

We need one last definition. As explained in Subsection \ref{subsec:general-idea}, one of our two mechanisms to ensure finiteness of the supertree is to treat as final results the partial derivations without an unvisited node labelled by a top nonterminal. We must however take into account the partial cycles that we can insert later and that could contain an unvisited top node. To express concisely what we will do, we define a graph $\mathcal G(\Gamma)$ for every set $\Gamma$ of partial cycles. Its vertices are all the pairs $(a',X)$ for $a' \in [0,A-1]$ and $X \in \mbox{Top}(G)$. There is an edge from $(a_1,X_1)$ to $(a_2,X_2)$ if there is an $(a_1,X_1)$-cycle in $\Gamma$ that contains $(a_2,X_2)$. We say that a partial derivation $\tau$ \emph{can reach an unvisited top node with the help of $\Gamma$} if there is a pair $(a_1,X_1) \in [0,A-1] \times \mbox{Top}(G)$ such that $a_1$ is the input and $X_1$ is the label of a node in $\tau$ with unspecified output, and $(a_1,X_1)$ can reach in $\mathcal G(\Gamma)$ a vertex $(a_2,X_2)$ such that there is an $(a_2,X_2)$-cycle in $\Gamma$ containing an unvisited top node.
\newline

We define the supertree in a smallest-fixpoint style. However, due to inconsequential nondeterminism, the supertree is not strictly the smallest structure satisfying a given property, but rather a minimal one. We will say ``the supertree'' throughout the section to mean ``one fixed supertree''. Definitions by smallest fixpoints are usually written in two parts: first, a list of basic objects that the set must contain (or the root of a tree), then a list of functions under which the set must be closed (or rules for adding children to a node). In our case, it is more convenient to formulate the definition in three parts: we also have a \emph{stop condition}. The rules for adding superchildren only apply at supernodes that don't satisfy the stop condition.
condition for adding superchildren are only required for supernodes that don't meet the stop condition. 


The supertree is a minimal tree $\mathcal T$, made of supernodes, which contains the \emph{basic supernode} and such that \emph{rules for adding superchildren} are satisfied at all the supernodes that don't satisfy the \emph{stop condition}.
\begin{itemize}
    \item \textbf{Basic supernode:} The superroot, which has neutral status
     and contains a single-node partial derivation. The label of the unique node in the partial derivation is $\mathcal S (G)$, and the input is $a$.
    \item \textbf{Stop condition:} We define a set $\Gamma(\mathcal T)$ of partial cycles. Consider all the failed supernodes $S$ of $\mathcal T$ (of partial derivation denoted $\tau$ and current node denoted $N$), such that $N$ has only its input specified, and has a strict ancestor $N'$ of same input and same label. For each such supernode, $\Gamma(\mathcal T)$ contains the subtree of $\tau$ rooted at $N'$. (This subtree is a partial cycle).
    
    We stop at a supernode $S$ if its partial derivation $\tau$ cannot reach an unvisited top node with the help of $\Gamma(\mathcal T)$ and if the current action in $\tau$ is not the first visit of a top node.
    
    \item \textbf{Rules for adding superchildren:} For every neutral supernode $S$, of partial derivation written $\tau$, if $S$ doesn't satisfy the stop condition, then the following hold:
\end{itemize}

\begin{enumerate}
    \item If the current action is the first visit of a node $N$ labelled by the terminal $c$, then
    \begin{itemize}
        \item If $\text{in}(N) + c \ge 0$, then $S$ has a superchild 
        whose partial derivation is a copy of $\tau$, where the output of $N$ is set to $\text{in}(N) + c$. Its status is successful if $\text{in}(N) + c \ge A + CD'$, and neutral otherwise.
        \item Otherwise, $S$ has a superchild with 
        the same partial derivation, but with status ``failed''.
    \end{itemize}
    \item If the current action is the first visit of a node $N$ labelled by the lower nonterminal $V$, then 
    \begin{itemize}
        \item If $G_V$ can cover $A + CD'$ starting from $\text{in}(N)$, then $S$ has a superchild with 
        the same partial derivation, but with successful status.
        \item Otherwise, for every $b \in \N$ such that $\text{in}(N) \xrightarrow{V} b$, $S$ has a superchild $S_b$ with neutral status
        whose partial derivation is a copy of $\tau$, where the output of $N$ is set to $b$.
    \end{itemize}
    \item If the current action is the first visit of a node $N$ labelled by the top nonterminal $X$, then $S$ has a superchild $S_{\alpha_1 \alpha_2}$
      for every rule $X \to \alpha_1 \alpha_2$ of $G$. Its partial derivation is a copy of $\tau$, where $N$ has a left child of same input as $N$, label $\alpha_1$, and unspecified output, and a right child of unspecified input and output, and of label $\alpha_2$. Its status is failed if $N-1$ has a strict ancestor of input $\text{in}(N_1)$ and label $\alpha_1$, and neutral otherwise.
    \item If the current action is the last visit of a left child, then $S$ has a superchild 
    whose partial derivation is a copy of $\tau$, where the input of the sibling of $N$ is set to $\text{out}(N)$. Its status is successful if the sibling of $N$ is labelled by a top variable and has an input greater than $A$, and neutral otherwise.
    \item If the current action is the last visit of a right child
    \begin{itemize}
        \item $S$ has a superchild 
        whose partial derivation is a copy of $\tau$, where the output of the parent of $N$ is set to $\text{out}(N)$. Its status is ``failed'' if the parent of $N$ has same input, label and output as one of its ancestor, and neutral otherwise.
        \item 
        $S$ has a superchild $S_{\mathcal C}$ for each equivalence class $\mathcal C$ of simple $(\text{in}(N),X)$-cycles satisfying the following two conditions:
        \begin{itemize}
            \item There is a failed supernode whose partial derivation's current node is the distinguished leaf of a simple $(\text{in}(N),X)$-cycle of class $\mathcal C$.
            \item For any representative $\gamma$ of $\mathcal C$, inserting $\gamma$ in $\tau$ doesn't introduce a repetition of pairs input-label among the strict ancestors of the current node.
        \end{itemize} 
        The partial derivation $\tau_C$ of $S_C$ is obtained by inserting the cycle $\gamma$ in $\tau$; and 
        $S_C$ has a neutral status.
    \end{itemize}
\end{enumerate}



\subsection{Size of the supertree} \label{subsec:finite_supertree}

In this subsection, we prove that the supertree is finite. This immediately implies that the supertree is computable, because coverability in 1-GVAS \cite{DBLP:conf/icalp/LerouxST15} and reachability in thin GVAS \cite{DBLP:conf/fsttcs/AtigG11} are decidable.
More precisely, we will prove that the size of the supertree is triply exponential.

\begin{lemma} \label{lem:size-supertree}
    The number of supernodes in the supertree is at most triply exponential in $\size(G)$. Moreover, all the partial derivations contained in the supertree are of size at most doubly exponential in $\size(G)$.
\end{lemma}

\begin{proof}

The arity of the supertree is at most doubly exponential in $\size(G)$. Indeed, the supernodes that can have the largest number of superchildren are those who create their superchildren with Rule 5. They have one superchild corresponding to no cycle insertion, and at most one superchild for each equivalence class of simple $(a',X)$-cycle, $(a',X)$ being the input and the label of the current node of their partial derivation. 

There only remains to bound the length of the superbranches and the size of the partial derivations. The same argument will give us both results; it relies on the key observation made in the following claim.

\begin{claim} \label{cl:depth-current-node}
    In each neutral supernode $S$ of the supertree, the partial derivation has no repeated input-label pairs among the ancestors of the current node. Consequently, the depth of the current node is at most $D$.
\end{claim}

\begin{proof}
    We prove the claim by induction of the depth of $S$. The partial derivation of the superroot has only one node, so obviously it has no repeated input-label pairs at all. 

    Now, suppose that the partial derivation $\tau$ of some supernode $S$ in the supertree has no repeated input-label pairs among the strict ancestors of the current node. We show that it is still the case for the partial derivation $\tau'$ of any neutral superchild $S'$ of $S$. If $S'$ is created by Rule 1, 2 or 4, then it has the same current node and current branch as $\tau$, so there is nothing to prove. If $S'$ is created by Rule 3, then the only new node added to the current branch doesn't introduce a repetition, since $S'$ has neutral status. If $S'$ is created by the first point of Rule 5, the current branch only becomes smaller, so no repetition can appear. Finally, if $S'$ is created by the second point of Rule 5, the constraint on cycle insertion explicitly forbids the introduction of a repetition.
\end{proof}

Take a superbranch, and let $(\tau_i)_{i \in I}$ be the sequence of partial derivations along this branch (at this point, $I$ can be of the form $[0,n]$ for some $n \in \N$ or equal to $\N$, since we have not proved that superbranches are finite). We will give a uniform bound on the number of nodes in any of the $\tau_i$ for $i \in I$. Note that $\tau_{i+1}$ is obtained from $\tau_i$ only by adding new nodes or defining previously unspecified counters; nodes are never deleted and specified counters are not changed. Moreover, there is always a change from $\tau_i$ to $\tau_{i+1}$, except when the status of the supernode $i+1$ is failed, in which case it is a superleaf. The length of the superbranch is thus at most three times the bound on the number of nodes in the $(\tau_i)_{i \in I}$, plus one.

We distinguish three kinds of nodes and bound their number separately.
\begin{itemize}
    \item Nodes where only the input value is specified: they are ancestors of the current node, which by Claim~\ref{cl:depth-current-node} is at depth at most $D$, so there are at most $D + 1$ of them.
    \item Nodes where neither the input nor the output is specified: they are right children of an ancestor of the current node, so there are at most $D + 1$ of them.
    \item Nodes where both input and output are defined: they are inside a subtree rooted at a left child of an ancestor of the current node. There are at most $D + 1$ such subtrees. Moreover, if a node has the same input, the same label and the same output as one of its ancestors, then the supernode corresponding to the last visit of this node has failed status (by the first point of Rule 5) and the superbranch ends at this supernode. If we put this supernode aside for simplicity's sake, all the $D + 1$ subtrees have depth at most $A \cdot \lvert \text{Top}(G) \rvert \cdot (A + CD')$, where $A, C$ and $D'$ are all exponential in $\size(G)$. All in all, the number of nodes with both input and output specified is doubly exponential in $\size(G)$.
\end{itemize}
Altogether, the number of nodes in the partial derivation $(\tau_i)_{i \in I}$ is at most doubly exponential, so the length of the superbranch is also at most doubly exponential. This gives a triply exponential bound on the number of supernodes in the supertree.
\end{proof}

As for time complexity, if $s$ is the size of the supertree, then computing it takes at most $s \cdot g(\sum_{V \in \low(G)} \size(H_V) + f(\size(G)))$  where $g$ is the time complexity of the reachability problem for thin 1-GVAS and $f$ is some exponential function.
Indeed, every time we call the thin reachability algorithm, it is on one of the thin 1-GVAS $H_V$ for $V \in \low(G)$ and with an input value at most exponential in $\size(G)$; moreover a new supernode is created at every such call. 
We also have at most one call per supernode to a coverability algorithm, on the same 1-GVAS. The complexity of such a call is also bounded by $g(\sum_{V \in \low(G)} \size(H_V) + f(\size(G)))$, since coverability is a special case of reachability.

\subsection{Proof of Lemma \ref{lem:small-lines} from the supertree} \label{subsec:lemma_from_supertree}

In this subsection, we use the supertree to prove Lemma \ref{lem:small-lines}. The approach diverges significantly depending on the existence of a successful supernode.


\subsubsection{Case where there is a successful supernode}

First, suppose that there is a successful supernode $S$, of partial derivation called $\tau$ and current node called $N$. There are three possibilities:
\begin{enumerate}
\item $N$ is labelled by a top nonterminal $X$ and has an input $a'$ greater than $A$. (This is the case if the supernode $S$ was created by Rule 4.)
\item $N$ has output greater than $A + CD'$. (This is the case if the supernode $S$ was created by Rule 1.)
\item $N$ is labelled by a lower nonterminal $V$ which, from its input, can cover $A + CD'$. (This is the case if the supernode $S$ was created by Rule 2.)
\end{enumerate}
In each case, we will apply Lemma \ref{lem:line-linear-set} to compute a semilinear representation of $R_G(a)$ from some threshold $T$, and then use Proposition \ref{prop:semilin-to-thin}. We start with Case 1. Recall that, by definition of $A$, there is an $X$-cycle $\gamma$ with positive left and global effect, of size exponential in $\size(G)$, whose left part is valid at $A$.
We remove the counters from $\tau$ and insert a simple derivation at every nonterminal leaf to obtain a complete derivation. Then, we insert at node $N$ enough copies of $\gamma$ to make the derivation valid at $a$. Call $\tau'$ the derivation obtained this way. We apply Lemma \ref{lem:line-linear-set} to the derivation $\tau'$ and the occurrence of $\gamma$ rooted at node $N$. We obtain $T$ such that $R_G(a) \cap [T, \infty) = (a + r + d\Z) \cap [T, \infty)$. Moreover, $T$ is exponential in $\size(G)$ and polynomial in the global effect of $\gamma, s_{\text{left}}$ and $s_{\text{right}}$, where $s_{\text{left}} = a' - a$ is the effect of the leaves of $\tau'$ on the left of the subtree rooted at $N$, and $s_{\text{right}}$ is the sum of the absolute values of the leaves on its right. Clearly, $s_{\text{right}}$ and the effect of $\gamma$ are exponential in $\size(G)$. The input $a'$ at $N$ is also exponential in $\size(G)$, because it is copied from an output counter, which was itself bounded by $A + CD'$. Hence, $s_{\text{left}}$ is also exponential in $\size(G)$, and this concludes the proof of Lemma \ref{lem:small-lines} in Case 1.
\newline

Now, we reduce Case 2 to Case 1. The idea is to propagate the large output value of $N$ until the input value of a top node. Recall that the stop condition, at a supernode $S'$ of partial derivation $\tau'$, intuitively says that we can still reach (the first visit of) a top node in the Euler Tour of $\tau'$. We had to formulate it in a particularly lengthy and intricate way (for the supertree to be easily computable), but here we only need the following simpler property, which the true stop condition immediately implies. Let $\mathcal G$ be the graph of a vertex set $[0,A-1] \times \text{Top}(G)$ and such that there is an edge from $(a_1,X_1)$ to $(a_2,X_2)$ if there is a simple $(a_1,X_1)$-cycle $\gamma$ containing, in its main branch, a node of input $a_2$ and of label $X_2$.
\newline

\noindent \textbf{Simple stop condition (at $S'$):} The current action is not the first visit of a top node. Additionally, there is no vertex $(a', X')$ in the graph $\mathcal G$ such that the following two properties hold
\begin{itemize}
  \item $(a', X')$ is reachable in $\mathcal G$ from the input-label pair of some strict ancestor of the current node
  \item there exists a simple $(a',X')$-cycle containing an unvisited top node.
\end{itemize}

This condition didn't hold at the superparent of $S$ (otherwise, $S$ would not exist). This implies that it doesn't hold at $S$ neither, because $S$ has the same partial derivation as its parent, except that the output the current node is now specified. Therefore, there is a path in $\mathcal G$ from the input-label pair $(a_1,X_1)$ of some ancestor of the current node, and to a vertex $(a_2,X_2)$ such that there is a simple $(a_2,X_2)$-cycle containing an unvisited top node. This path can be taken of length at most $A \cdot \lvert \text{Top}(G)\rvert$. Thus, by inserting at most $A \cdot \lvert \text{Top}(G)\rvert$ simple partial cycles in the partial derivation $\tau$ of $S$, we can construct a partial derivation $\tau'$ where the current node has depth at most $D' = D + (A \cdot \lvert \text{Top}(G)\rvert)^2$ and where there is an unvisited top node $M$.

To finish the reduction, we continue the Euler Tour of $\tau'$ from its current node until the first visit of $M$, but with new rules: we propagate counter values and, every time we encounter a leaf labelled by a nonterminal, we insert a derivation of size exponential in $\size(G)$, of effect bounded by $C$ in absolute value, and valid at $C$. (Such derivations exist by definition of $C$.) This gives us a partial derivation $\tau''$ as in Case 1. The input of the current node $M$ is at most $b$ plus some exponential function of $\size(G)$. Thus, for complexity bounds, we only need to ensure that $b$ is not too large. It is indeed the case: the value $b$ is of the form $a + c$, where $c$ is a terminal of $G$ and $a$ is the input of a leaf labelled by $c$ in the partial derivation of the supernode $S$. In turn, $a$ was copied from the input of a parent node, which is labelled by a top nonterminal, and therefore is bounded by $A$ (otherwise, the grandparent of $S$ would be a successful supernode, and $S$ would not exist).
\newline


Finally, there is an easy reduction from Case 3 to Case 2: apply the thin reachability algorithm to see whether, from input of the current node $N$, in the thin 1-GVAS $H_V$, we can reach $A + CD'$, or $A + CD' + 1$, or $A + CD' + 2$, ... until you find a reachable value $b$. You will eventually find one, because $A + CD'$ is coverable, but it is difficult to control its size. We will adopt two very different strategies to avoid this issue, depending on whether $G$ has a negative cycle or not.

\paragraph*{If $G$ has no negative cycle:}

We do the reduction as described above. 
The following lemma ensures that $b$ is of size exponential in $\size(G)$.

\begin{lemma} \label{lem:1GVAS-without-negative-cycle}
    Let $G$ be a 1-GVAS without negative cycles. Suppose that $G$ can cover the value $B \in \N$ from an input value $a < B$. Then, there is $b \in \N$ of size at most polynomial in $B$ and exponential in $\size(G)$ such that $a \xrightarrow{G} b$ and $B \le b$.
\end{lemma}

The proof of Lemma~\ref{lem:1GVAS-without-negative-cycle} is technical and independent from the main reasoning, so we defer it to Appendix~\ref{sec:app2}.

\paragraph*{If $G$ has a negative cycle:}
In this case, we can compute a small value $T$ such that $R_G(a) \cap [T, \infty) = (a + r + d\Z) \cap [T, \infty)$, independently of the values $b \ge A + CD'$ that the thin 1-GVAS $H_V$ can reach from input $a'$. The idea is that, even if those values $b$ are very large, we can insert negative cycles later to obtain derivations of small size.

We know that there is $b \ge A + CD'$ such that $G_V$ can reach $b$ from the input of $N$ (but we don't compute it.) If we apply the reduction from Case 2 to Case 1 to a copy of $\tau$ where the output of $N$ is set to $b$, then we obtain a partial derivation $\tau'$, produced by $G$, whose current node $M$ has input $a' \ge A$.
Then, we can do what we did in Case 1 (insert at every nonterminal leaf of $\tau'$ a derivation of effect at most $C$ in absolute value, and valid at $C$; insert at node $M$ a cycle $\gamma$ with positive left and global effect, of size exponential in $\size(G)$, whose left part is valid at $A$; then insert more copies of $\gamma$ at $M$ until the complete derivation is valid from input $a$) to obtain a complete derivation $\theta$.
Note that we only proved the existence of $\theta$, we didn't compute it. From its existence, we deduce the existence of a similar complete derivation $\theta'$, but with $s_{\text{left}}$ exponential in the size of $G$. More precisely, $\theta'$ is obtained from $\theta$ by inserting an $X$-cycle $\gamma_{X,-}$ of negative effect at $M$, which we describe now.

The cycle $\gamma_{X,-}$ will be constructed by taking derivations $D: X \Rightarrow u_1 X u_2 X u_3$ and $D': X \Rightarrow v$ with $u_1,u_2,u_3,v \in \Z^*$ of size exponential in $\size(G)$, inserting $D'$ on the left in $D$, and in between many copies of a cycle $\gamma_{-}$ of negative effect, which remain to choose. The goal is that the left part of the cycle $\gamma_{X,-}$ decreases the input of $M$ from something close to $b$, which is arbitrary large, to something slightly larger than $A$. (``Close to'' and ``slightly larger'' here mains ``within a distance exponential in $\size(G)$''.) For that, we need the effect of $\gamma_{-}$ to be exponential in $\size(G)$, so that we can adjust precisely the effect of $\gamma_{X,-}$ by inserting the appropriate number of copies of $\gamma_{-}$. We also need the nonterminals of $\gamma_{-}$ to be split in an appropriate way between its left and right part. Imagine for instance that we had a $\gamma_{-}$ of the form 
$V: \Rightarrow -10 V 9$ for some nonterminal $V$. In that case, $\gamma_{X,-}$ could not decrease the input value at $M$ by more than (approximately) 10\% without making the counters fall below zero. We will thus choose $\gamma_{-}$ with a negative global effect, but a nonnegative left effect. In the next paragraph we show that such a cycle always exists in our situation.

By hypothesis, $G$ has negative cycles. It also has negative simple cycles, as all cycles can be decomposed into simple cycles. Let $\gamma_{\text{simple},-}: V \Rightarrow u V v$ be a negative simple cycle ($V$ can be a top or a lower nonterminal, $u, v \in \Z^*$, and $\sum u + v < 0$).
We reuse the derivations $D: X \Rightarrow u_1 X u_2 X u_3$. We also need new derivations $D_1: X \Rightarrow v_1 V v_2$ and $D_2: V \Rightarrow w$ with $v_1,v_2,w \in \Z^*$ of size exponential in $\size(G)$.
(We don't really need them if $V = X$, but they don't hurt, so we don't treat the case $V=X$ separately.)
This time, we first insert a copy of $D$ into another copy of $D$, so that we obtain a derivation with three $X$ leaves. In the rightmost one, insert the block formed by $D_2$ inserted into $D_1$, with $n$ copies of $\gamma_{\text{simple},-}$ in between. Insert a similar block in the leftmost $X$-leaf, but replacing the copies of $\gamma_{\text{simple},-}$ by $m$ copies of a cycle of positive left effect, positive global effect, and size exponential in $\size(G)$ (such a cycle exist, by condition 3 of Proposition \ref{prop:branching-top-component}). For $m$ and $n$, choose the least values making the left part of $\gamma_{-}$ nonnegative, and the right part negative.




\subsubsection{Case where there is no successful supernode} \label{subsubsec:no-success}

In this paragraph, we suppose that the supertree has no successful supernode. Recall that our goal is to use the supertree to prove Lemma~\ref{lem:small-lines}, namely to find a thin 1-GVAS $H$ such that $H \under_{\{a\} \times [T,\infty)} G$ for some $T \in \N$. Actually, in this case, we will have $T=0$. We have at our disposal a thin 1-GVAS $H_V$ equivalent to $G_V$ for every lower nonterminal $V$. To simplify the expressions, we suppose that the initial nonterminal of $G_V$ is also $V$, and that the sets of nonterminals of the $H_V$ are pairwise disjoint.

The thin 1-GVAS $H$ is obtained by adding, with the help of the superleaves, a thin component on top of the $H_V$. To provide some intuition before the formal definition of $G$, we start with a simplified but incorrect version, and we modify it progressively towards the correct definition. In the following, the symbol $\uplus$ is used rather than $\cup$ to stress that we take the union of disjoint sets.
\begin{itemize}
    \item A natural idea would be to merge all the $H_V$, and add a new initial nonterminal $\mathcal S(H)$ with a rule $\mathcal S(H) \to \alpha$ for every $\alpha$ which is the yield of the partial derivation of some neutral superleaf. We can do this because neutral superleaves are supernodes where the stop condition apply, so they only contain lower nonterminals. Clearly, with this definition, $H$ is thin. Moreover, partial derivations contained in some supernode of the supertree are derivations produced by $G$, so $H$ under-approximates $G$.
    \item However, with the above definition, we don't have $R_H(a) \subseteq R_G(a)$, even from some threshold. Indeed, superleaves capture all derivations with counters produced by $G$ from input $a$, but only if we take into account the cycles of left effect zero that could have been inserted if we had finished the Euler Tour. A solution is to add, for every $a' \in [0,A-1]$ and every $X \in \topp(G)$ a new nonterminal $V_{a',X}$ allowing to insert the right part of $(a',X)$-cycles. For a simple $(a_0,X_0)$-cycle as represented on Figure~\ref{fig:def-H}, we would have a rule $V_{a_0,X_0} \to \beta_1 \, V_{a_1,X_1} \, V_{a_2,X_2}  \, \beta_3  \, V_{a_3,X_3}  \, V_{a_4,X_4}$. The initial rules associated to a superleaf $S$ would be $\mathcal S(H) \to \beta_1  \, V_{a_1,X_1}  \, \beta_2  \,  V_{a_2,X_2} \cdots$ where $a_1, X_1$ are the input and label of the parent of the current node, and $\beta_1$ is the label of its right child if this right child is out of the current branch, and $0$ otherwise; $a_2, X_2$ and $\beta_2$ are defined similarly for the grandparent of the current node.
    \item Unfortunately, the nonterminals $V_{a',X}$ are not necessarily thin. To enforce thinness, we insert cycles following the same discipline as in the definition of the supertree. Namely, we have a nonterminal $V_{a',X,F}$ for every $a' \in [0,A-1]$, every $X \in \topp(G)$ and every $F \subseteq [0,A-1] \times \topp(G)$. $F$ is the set of \emph{forbidden nonterminals}; $V_{a',X,F}$ insert $(a,X)$-cycles that don't contain them and that also enforce the discipline. The example on Figure~\ref{fig:def-H} becomes $V_{a_0,X_0,F} \to \beta_1 V_{a_1,X_1,F\uplus \{(a_2,X_2), (a_3,X_3), (a_4,X_4)\}} V_{a_2,X_2,F \uplus \{(a_3,X_3),(a_4,X_4)\}} \beta_3 V_{a_3,X_3,F \uplus \{(a_4,X_4)\}} V_{a_4,X_4,F}$.
\end{itemize}

\begin{figure}[H]
    \centering
    \begin{tikzpicture}[scale = 0.6]
    \definecolor{mycolor}{rgb}{0,0.7,0}
            \node  (0) at (0, 5) {$N_4 : (a_4, X_4, \bot)$};
            \node  (1) at (-2.5, 3.25) {};
            \node  (2) at (2.25, 2.75) {$N_3 : (a_3, X_3,\bot)$};
            \node  (3) at (0, 0.5) {$N_2 : (a_2, X_2, \bot)$};
            \node  (4) at (4.5, 0.5) {$\beta_3$};
            \node  (5) at (-2.5, -1.25) {};
            \node  (6) at (2.25, -1.25) {$N_1 : (a_1, X_1, \bot)$};
            \node  (7) at (-0.25, -3.5) {$N_0 : (a_0, X_0, b_0)$};
            \node  (8) at (4.75, -3.5) {$\beta_1$};
            \node  (9) at (-3.75, 1.75) {};
            \node  (10) at (-1.5, 1.75) {};
            \node  (11) at (-3.5, -2.25) {};
            \node  (12) at (-1.75, -2.25) {};
            \draw (0) to (1.center);
            \draw (0) to (2);
            \draw (2) to (3);
            \draw (2) to (4);
            \draw (3) to (6);
            \draw (3) to (5.center);
            \draw (7) to (6);
            \draw (6) to (8);
\fill[mycolor] (1.center) -- (9.center) -- (10.center) -- cycle;

\fill[mycolor] (5.center) -- (11.center) -- (12.center) -- cycle;
    \end{tikzpicture}
    \caption{This represents a partial derivation of current node $N_0$ (and an $(a_0, X_0)$-cycle if $(a_0,X_0) = (a_4,X_4)$). The notations represented there are the same as in the definitions. The green triangles represent subtrees where input and output are specified everywhere. What only matters about those trees is their yields.}
    \label{fig:def-H}
\end{figure}
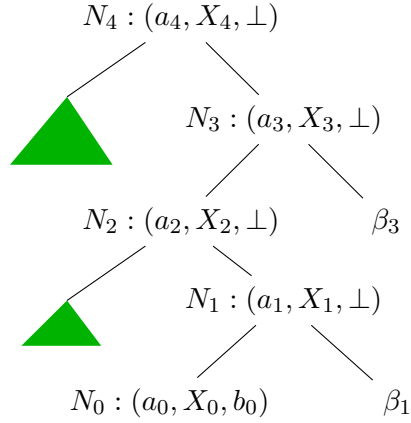

Now we define $H$ more formally.

\begin{itemize}
    \item The nonterminals of $H$ are the nonterminals of all the $H_V$, $V \in \low(G)$, a fresh initial nonterminal $\mathcal S(H)$, and a nonterminal $V_{a',X,F}$ for every $a' \in [0,A-1]$, every $X \in \topp(G)$ and every $F \subseteq [0,A-1] \times \topp(G)$.
    \item $H$ has all the production rules of the $H_V$.
    \item We say that a set $F \subseteq [0,A-1] \times \topp(G)$ and an equivalence class $\mathcal C$ of simple partial cycles are \emph{compatible} if the main branch of any representative of $\mathcal C$ doesn't have an input-label pair in $F$.
    For every $a_0 \in [0,A-1]$, every $X_0 \in \topp(G)$ and every $F \subseteq [0,A-1] \times \topp(G)$, for every equivalence class $\mathcal C$ compatible with $F$, 
    there is a corresponding rule. To define it, let $N_0,N_1,...,N_k$ be the sequence of nodes on the main branch of any representative of $\mathcal C$, read from bottom to top. For $i \in [1,k]$, let $\beta_i$ be the  label of the right child of $N_i$ if it is out of the main branch, and $0$ otherwise. We also write $a_i$ for the input of $N_i$ and $X_i$ for its label, this time for $i \in [0,k]$ (so that $a_0,X_0 = a_k, X_k)$.
For readability, define for a moment $Y_i := V_{a_i,X_i,F \uplus F_i}$ for all $i \in [1,k]$.
    The corresponding rule is $V_{a_0,X_0,F} \to \beta_1 Y_1 \beta_2 Y_2 \dots \beta_k Y_k$
    where, for all $i \in [1,k]$, $F_i := \{(a_j,X_j) \mid i < j \le k \}$.
    \item For every neutral superleaf $S$ of partial derivation $\tau$, there is a corresponding initial rule. We use the same notations $k$, $a_i, X_i$ and $\beta_i$ as above. 
    Decompose the yield of $\tau$ into $\alpha \beta_1 ... \beta_k$ (hence $\alpha$ is the word formed by the labels of the leaves of $\tau$ whose input and output is specified). For readability, define for a moment $Z_i := V_{a_i,X_i,F_i}$ for all $i \in [1,k]$.
    The initial production rule associated to $S$ is $\mathcal S(H) \to \alpha \beta_1 Z_1 \beta_2 Z_2 \dots \beta_k Z_k$ with the same $F_i$ as above.
\end{itemize}

It is clear that $H$ is thin. Moreover, the size of what is added to the $H_V$ to obtain $H$ is triply exponential in $\size(G)$. Indeed
\begin{itemize}
    \item The number of new nonterminals is at most doubly exponential (because $A$ is at most exponential in $\size(G)$.)
    \item Every $V_{a',X,F}$ is the left-hand side of doubly-exponentially many rules, each of size at most exponential.
    \item Similarly, there are at most triply-exponentially many initial rules (one for each superleaf), each of size at most doubly exponential.
    This is because the superbranches have at most doubly exponential depth, so the partial derivations contained in superleaves have at most doubly-exponentially many leaves.
\end{itemize}

\begin{lemma} \label{lem:correction-mainproof}
    $H \under_{\{a\} \times \N} G$
\end{lemma}


\begin{proof}

    All the partial derivations contained in some supernodes are, if their counters are removed, a derivation produced by $G$. Therefore, the inclusion $R_{H} \subseteq R_G$ is immediate. The other inclusion, $R_{G}\vert_{\{a\} \times \N} \subseteq R_{H}$, is harder to establish. We prove it by describing an algorithm.
    It takes as input a derivation with counters $\theta$ which is produced by $G$, has input $a$ and has no cycle of left and right effects both equal to zero. We call $b$ its output. The algorithm returns a superleaf $S$ whose associated production rule in $H$ is of the form $\mathcal S (H) \to \alpha$, such that there is $w \in \Z^*$ which is a valid run from $a$ to $b$ and with $\alpha \Rightarrow w$.
    
    The algorithm performs simultaneously an Euler Tour in $\theta$ and a walk through the supertree. Thus, at each step, we have a current action in $\theta$ and a current supernode. At some points in the double traversal, instead of moving to the next action in the Euler Tour of $\theta$, we will extract cycles from $\theta$ for later reinsertion. In general, removing a cycle from a derivation with counters breaks the flow conditions, but here we only remove cycles with left effect zero and reinsert them at the last visit of the root. Consequently, the Euler Tour will never see the temporary breaches of flow conditions.

    We will maintain the following invariant between the current node $M$ of the Euler Tour in $\theta$, and the current supernode $S$, of partial derivation $\tau$ and current node $N$:
    \begin{itemize}
        \item $M$ and $N$ have the same input. Moreover, the current action in the Euler Tour of $\tau$ is the last visit of $M$ if and only if the current action of $\tau$ is the last visit of $N$. In this case $M$ and $N$ have the same output.
        \item The current branches of $\theta$ and $\tau$ induce the same sequence of inputs, labels and of (potential) labels of right children out of the current branch.
        \item The status of $S$ is neutral.
    \end{itemize}
    
    The algorithm initiates the traversal at the first visit of the root of $\theta$ and at the superroot. It terminates when a superleaf is reached and returns this superleaf. 

    We describe below which moves are allowed in the double traversal.
    This description follows naturally the rules that define the supertree. To state them easily, we call $M$ and $\mbox{Act}$ the current node and current action in the Euler Tour of $\theta$, $S$ the current supernode, $\tau$ its partial derivation, $N$ the current node of $\tau$, and $a$ the input value common to both $M$ and $N$ (by the invariant). 
    \begin{enumerate}
    \item If $\mbox{Act} = \text{first}(M)$ (which implies with the invariant that the current action in $\tau$ is $\text{first}(N)$), and if $M$ and $N$ are labeled by a terminal $c$, we move to the unique superchild of $S$, which has neutral status because $a + c$ is the output value of $M$ and therefore nonnegative. We also move to the next action in the Euler Tour.
    \item If $\mbox{Act} = \text{first}(M)$, and if $M$ is labeled by a lower nonterminal $V$, then we know that the output of $M$ is bounded by $A + CD'$, because we assume that there is no successful supernode. We choose as next supernode the superchild of $S$ that sets the output of $N$ to the same output value as $M$. We also move to the next action in the Euler Tour.
    \item If $\mbox{Act} = \text{first}(M)$, and if $M$ is labeled by a top nonterminal $X$
    \begin{itemize}
        \item If $M$ has an ancestor $M'$ with the same input and the same label, then we store a copy of the cycle composed of the descendants of $M'$ that are not strict descendants of $M$. 
        Moreover, we remove this cycle from $\tau$, which means that all the descendants of $M'$ which are not descendants of $M$ are suppressed, and that $M$ becomes a child of the former parent of $M'$. We don't advance in the Euler Tour, but we go back to the last supernode associated to a first visit of $M'$.
        \item Otherwise, we advance in the Euler Tour and we move to the superchild corresponding to the production rule applied to $M$ in $\tau$.
    \end{itemize}
    \item If $\mbox{Act} = \text{last}(M)$ and if $M$ is a left child, we move to the next action in the Euler Tour and to the unique superchild of $S$. (Note that the second point of the invariant implies that $N$ is also a left child.)
    \item If $\mbox{Act} = \text{last}(M)$ and if $M$ is a right child (which implies with the invariant that $N$ is also a right child)
    \begin{itemize}
        \item If there is no stored cycle with special leaf $M$, then we move to the next action in the Euler Tour and to the superchild of $S$ corresponding to no cycle insertion.
        \item Otherwise, retrieve the cycle $\gamma$ with distinguished leaf $M$ in the storage (there can be at most one such cycle due to how they are extracted in Rule 3) and reinsert it in $\tau$, above $M$ (that is: merge the two nodes $M$, then remove the edge between $M$ and its parent $M'$ in $\tau$. The new child of $M'$ is the root of $\gamma$.) We don't advance in the Euler Tour, but we move to the superchild $S_{\mathcal C}$, where $\mathcal C$ would be the equivalence class of $\gamma$ if we turned it into an $(a,X)$-cycle (\emph{i.e.}, if we cut the trees on the right of the main branch and made the counter values after the input of the current node undefined). We prove the existence of $S_{\mathcal C}$.
    \end{itemize}
    \end{enumerate}
    We prove that $S_{\mathcal C}$ always exists in the second point of Rule 5 above, namely, that our constraints on cycles insertion allow for this insertion. Observe that the cycles we extract could actually have been stored on a pushdown. Indeed, if a cycle is extracted above node $M$ at its first visit, then it is inserted again at its last visit. Between those two moments, a complete Euler Tour of the subtree rooted at $M$ is performed, so all the cycles that have been extracted after the cycle above $M$ are reinserted before the last visit of $M$. Therefore, when the cycle above $M$ is reinserted at the last visit of $M$, the branch from the root to $M$ is the same as it were at the first visit of $M$. It can then be shown by induction that the branch from the root to the parent of $M$ contains no repetition of pairs input-label.

    These five rules cover all cases except for when $\text{Act}$ is the last visit of the root. However, the invariant ensures that upon reaching the last visit of the root in the Euler tour, the corresponding supernode is a superleaf; in this case, we stop and return this superleaf. Consequently, the double traversal is never blocked. Furthermore, it terminates because, at every move of the double traversal, the Euler Tour of $\theta$ either advances or stays where it is. It doesn't advance in the first point of Rule 3 and in the last point of Rule 5, which cannot happen twice in a row.

    There only remains to prove that the production rule in $H$ associated to the returned superleaf $S$ is of the form $\mathcal S (H) \to \alpha$, such that there is $w \in \Z^*$ which is a valid run from $a$ to $b$ and with $\alpha \Rightarrow w$. Let $\theta'$ be the derivation with counters and $M$ be the current node of its Euler Tour at the end of the algorithm. Note that $\theta'$ is in general different than the $\theta$ we take as input because some cycles could have been not reinserted yet. These cycles are those for which we haven't done the last visit of the distinguished leaf, as well as those whose distinguished leaf doesn't appear in $\theta'$ (but in some other extracted cycle).

    The rule corresponding to $S$ in $H$ is $\mathcal S(H) \to \alpha \beta_1 V_{\alpha_1, X_1, F_1} \beta_2 V_{a_2,X_2, F_2}\dots V_{a_k,X_k,F_k}$ where $\alpha$ is the word formed by the labels of the leaves of $\tau$ whose input and output are both specified (and the $a_i$, $X_i$ and $F_i$ are the same notations as in the definition of $H$.)
    By the invariant, $\alpha \Rightarrow v$ where $v$ is a valid run from $a$ to $\text{in}(N) = \text{in}(M)$. The invariant also implies that $\theta'$ and $\tau$ have the same current branch and the same right children of the nodes of the superbranch. Therefore, it is not difficult to see that $\beta_1 V_{\alpha_1, X_1, F_1} \beta_2 V_{a_2,X_2, F_2}\dots V_{a_k,X_k,F_k} \Rightarrow w''$ where $w''$ is a valid run from $\text{in}(M)$ to $b$ (just simulate with the rules of $H$ the end of the Euler Tour in $\theta'$; there are only cycle reinsertions and no further cycle extractions, as there are no more first visits of nodes labelled by top nonterminals.)

    \end{proof}


\section{Proof of Lemma~\ref{lem:bounded-area}}\label{sec:bounded-area}

Let $G$ be a top-branching 1-GVAS and $B \in \N$. We suppose that we have 
\begin{itemize}
\item  for each $V \in \low(G)$, a thin 1-GVAS $H_V$ of same reachability relation as $G_V$
\item  for each $X \in \topp(G)$ a thin 1-GVAS $H_X$ which $S$-exactly
under-approximates $G_X$ for $S = \N^2 \setminus [0,B]^2$.
\end{itemize}
We want to compute a thin 1-GVAS $H$ equivalent to $G$.
Similarly as in Section~\ref{subsubsec:no-success}, we suppose that for every $W \in \mathcal N(G)$, the initial nonterminal of $H_W$ is also called $W$. Apart from that, $H_W$ has no common nonterminal with $G$. Moreover, the sets of nonterminals of the 1-GVAS $H_W$ are pairwise distinct. This will considerably simplify the expressions below.





Our strategy is the following: consider a derivation with counters witnessing that $a \xrightarrow{G} b$ for some $a,b \in \N$. Suppose that it contains no cycle of left and right effects equal to zero. Since there is no repetition of input, output and label, every branch contains before depth $(B+1)^2 \lvert \text{Top}(G)\rvert$ a node labelled by a terminal, a lower nonterminal, or a top nonterminal with an input or an output exceeding $B$. For every $V$-derivation with counters produced by $G$, where $V \in \low(G)$, there is a derivation with counters produced by $H_V$ with the same input and output (because $H_V$ and $G_V$ have the same reachability relation). Similarly, for every $X$-derivation with counters produced by $G$, where $X \in \topp(G)$, if its input or output counter exceeds $B$, then there is a derivation with counters produced by $H_X$ with same input and output. 
Therefore, top nonterminals of $G$ are only required in a prefix of bounded depth; below, they can be simulated by their thin under-approximations. The thin 1-GVAS $H$ that we construct guesses this prefix in one rule, so that it only needs the thin under-approximations of top nonterminals of $G$, and the thin equivalent of lower nonterminals.

More precisely, we define $H$ by
\begin{align*}
    \mathcal S(H) & := S \\
    \mathcal N(H) & := \{S\} \uplus \biguplus_{W \in \mathcal N(G)} \mathcal N(H_W) \\
    \mathcal P(H) & := \{ S \to \alpha \mid \alpha \text{ is the yield of a derivation of } G \text{ of depth at most } (B+1)^2 \lvert \mathcal N(G) \rvert \} \\
    & \uplus \biguplus_{W \in \mathcal N(G)} \mathcal P(H_W)
\end{align*}
(where $\uplus$ is used to emphasise that we are taking the union of disjoint sets).


We claim the following.

\begin{claim}\label{cl:g-h-equivalence}
The 1-GVAS $G$ and $H$ have the same reachability relation.
\end{claim}

Notice that the size of $H$ is bounded by $\sum_{V \in \low(G)}\size(H_V) + \sum_{X \in \topp(G)}\size(H_X)$
plus the total size of all the yields of derivations of $G$ of depth at most $C = (B+1)^2 \Nn(G)$.
A yield of a single such derivation is of length at most $2^C$, as derivations are binary.
The number of possible derivations can be doubly-exponential,
but is bounded by some doubly-exponential function of $B$ and $\size(G)$.
Thus, the total size of all the yields of derivations of $G$ of depth $C$ is also at most doubly-exponential in $B$ and $\size(G)$.
Therefore
\[
\size(H) \leq \sum_{V \in \low(G)}\size(H_V) + \sum_{X \in \topp(G)}\size(H_X) + f(B + \size(G))
\]
for some doubly-exponential function $f$, as required in Lemma~\ref{lem:bounded-area}.
Also iterating through all these doubly-exponentially many derivations and constructing the rules of $H$ takes
at most doubly-exponential time in $B$ and $\size(G)$. Therefore $H$ satisfies the bounds for the size and the construction
time of Lemma~\ref{lem:bounded-area}. Thus, to finish the proof of Lemma~\ref{lem:bounded-area} it is enough to show Claim~\ref{cl:g-h-equivalence}.

\begin{proof}[Proof of Claim~\ref{cl:g-h-equivalence}]
In order to show that $\reach{H} = \reach{G}$ we show two inclusions. \\

\paragraph*{Inclusion $\reach{H} \subseteq \reach{G}$.}
Let $\tau_H$ be a $H$-derivation with counters.
We aim at constructing a $G$-derivation with counters $\tau_G$ such that $\tau_G$ and $\tau_H$ have the same input and output.
Let $S \to \alpha$ be the production rule applied at the root of $\tau_H$. 
To construct $\tau_G$, start from a derivation produced by $G$ and of yield $\alpha$. (The existence of such a derivation follows immediately from the definition of $\Pp(G)$.) Then, for every child $N$ of $S$ in $\tau_H$, find a derivation with counters $\sigma$ produced by $G$ with same input, output and label at the root as the subtree rooted at $N$ in $\tau_H$. Such a derivation exists, because for all $W \in \Nn(G)$, $H_W$ under-approximates $G_W$. Insert it in $\tau_G$ at the node corresponding to $N$, and propagate the counters values everywhere. \\

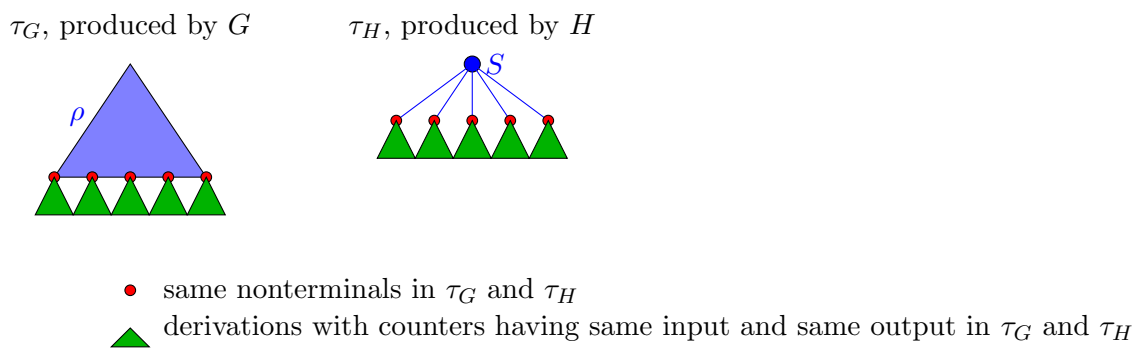
\begin{figure}[H]
    \begin{tikzpicture}
    
  	\definecolor{mycolor}{rgb}{0,0.7,0}

        \node at (3, 5.5) {$\tau_G$, produced by $G$};
    
        
        \filldraw[fill=blue!50, draw=black] (3,5) -- (2,3.5) -- (4,3.5) -- cycle;
        \node[blue] at (2.3, 4.3) {\(\rho\)};
        
        \draw[fill=red] (2,3.5) circle (2pt);  
        \draw[fill=red]   (2.5,3.5) circle (2pt);  
        \draw[fill=red] (3,3.5) circle (2pt);  
        \draw[fill=red] (3.5,3.5) circle (2pt);  
        \draw[fill=red]   (4,3.5) circle (2pt);  
        
        \filldraw[fill=mycolor, draw=black] (2,3.5) -- (1.75,3) -- (2.25,3) -- cycle;  
        \filldraw[fill=mycolor, draw=black] (3,3.5) -- (2.75,3) -- (3.25,3) -- cycle;  
        \filldraw[fill=mycolor, draw=black] (3.5,3.5) -- (3.25,3) -- (3.75,3) -- cycle;  
        
        \filldraw[fill=mycolor, draw=black] (2.5,3.5) -- (2.25,3) -- (2.75,3) -- cycle;  
        \filldraw[fill=mycolor, draw=black] (4,3.5) -- (3.75,3) -- (4.25,3) -- cycle;  
        
        \node at (7.5, 5.5) {$\tau_H$, produced by $H$};
    
    
        \filldraw[fill=blue] (7.5,5) circle (3pt);  
        \node[blue] at (7.8, 5) {$S$};  
    
        \draw[draw=blue] (7.5,5) -- (6.5,4.25);  
        \draw[draw=blue] (7.5,5) -- (7,4.25);  
        \draw[draw=blue] (7.5,5) -- (7.5,4.25);  
        \draw[draw=blue] (7.5,5) -- (8,4.25);  
        \draw[draw=blue] (7.5,5) -- (8.5,4.25);  
    
        \draw[fill=red] (6.5,4.25) circle (2pt);  
        \draw[fill=red]   (7,4.25) circle (2pt);  
        \draw[fill=red] (7.5,4.25) circle (2pt);  
        \draw[fill=red] (8,4.25) circle (2pt);  
        \draw[fill=red]   (8.5,4.25) circle (2pt);  
    
        \filldraw[fill=mycolor, draw=black] (6.5,4.25) -- (6.25,3.75) -- (6.75,3.75) -- cycle;  
        \filldraw[fill=mycolor, draw=black] (7.5,4.25) -- (7.25,3.75) -- (7.75,3.75) -- cycle;  
        \filldraw[fill=mycolor, draw=black] (8,4.25) -- (7.75,3.75) -- (8.25,3.75) -- cycle;  
    
        \filldraw[fill=mycolor, draw=black] (7,4.25) -- (6.75,3.75) -- (7.25,3.75) -- cycle;  
        \filldraw[fill=mycolor, draw=black] (8.5,4.25) -- (8.25,3.75) -- (8.75,3.75) -- cycle;  
    
        \draw[fill=red] (3,2) circle (2pt);
        \node[anchor=west] at (3.3,2) {same nonterminals in $\tau_G$ and $\tau_H$};
    
    
        \filldraw[fill=mycolor, draw=black] (3,1.5) -- (2.75,1.25) -- (3.25,1.25) -- cycle;
        \node[anchor=west] at (3.3,1.5) {derivations with counters having same input and same output in $\tau_G$ and $\tau_H$};
    
    
    \end{tikzpicture}
    \label{fig:sec6}
    
\caption{Correspondence between derivations in $G$ and $H$}
\end{figure}

\paragraph*{Inclusion $\reach{G} \subseteq \reach{H}$.}
Let $\tau_G$ be a $G$-derivation with counters. Similarly as above,
we aim at constructing a $H$-derivation with counters $\tau_H$ such that $\tau_G$ and $\tau_H$ have the same input and output.

Identify in $\tau_G$ a prefix $\rho$ of depth at most $(B+1)^2 \lvert \mathcal N(G)\rvert$, where every leaf of $\rho$ is either a terminal, a lower nonterminal, or a top nonterminal with an input or output exceeding $B$. Such a prefix exists, because no branch contains two nodes with the same inputs, labels and outputs. The derivation $\tau_H$ can be constructed by first applying the rule $S \to \alpha$ and then simulating the remaining part of $\tau_G$. More precisely, for every leaf $N$ of $\rho$ labelled by a nonterminal $W$, find a derivation with counters produced by $G$ with the same inputs, outputs and labels at the root of the subtree rooted at $N$ in $\tau_G$. Such a derivation exists, because for all $W \in \mathcal N(G)$, $H_W$ $\N^2 \setminus [0,B]^2$-exactly under-approximates $G_W$. Insert it in $\tau_G$ at the node corresponding to $N$, and propagate the counters values everywhere.

On Figure~\ref{fig:sec6} we illustrate the correspondence between derivations produced by $G$ and $H$.

\end{proof}


\section{Future research}\label{sec:future}

\paragraph*{Complexity of reachability in 1-GVAS}
It is natural to ask about the complexity of the reachability problem for 1-GVAS.
The best known lower bound for the problem is \pspace-hardness shown in~\cite{DBLP:journals/ipl/EnglertHLLLS21}.
Our work reduces the reachability problem to the thin 1-GVAS case, but unfortunately in~\cite{DBLP:conf/fsttcs/AtigG11}
no complexity upper bound is provided. The problem for thin 1-GVAS with $k$ nonterminals
is reduced there to the reachability problem for $(2k)$-dimensional VASS with nested zero-tests, which definitely have very high complexity
(it is at least \ackermann-hard as VASS reachability problem~\cite{DBLP:conf/focs/Leroux21,DBLP:conf/focs/CzerwinskiO21}).
Therefore, indeed there seems to be no hope to get a good complexity upper bound this way.

However, it is reasonable to conjecture that if there is a derivation $a \trans{X} b$ for some $s, t \in \N$
and a thin 1-GVAS nonterminal $X$ then there is also a small derivation. Actually, to our best knowledge,
in all known examples of thin 1-GVAS if there is a derivation then there is also a derivation of at most exponential size.
We therefore state the following conjecture.

\begin{conjecture}\label{conj:small-derivation}
For each thin 1-GVAS $G$ is there is a derivation of $G$ from $a$ to $b$, for some $a, b \in \N$
then there is also a derivation of $G$ from $a$ to $b$, which is of size at most exponential in size of $G$,
and size of the binary encodings of numbers $a$ and $b$.
\end{conjecture}

Conjecture~\ref{conj:small-derivation} would immediately imply an exponential space algorithm
for the reachability problem in thin 1-GVAS. In particular, by Corollary~\ref{cor:thin-elementary},
this would deliver an elementary complexity for the reachability problem for 1-GVAS.

\paragraph*{Higher dimensions of GVAS}
Clearly, after providing decidability of the reachability problem for one-dimensional GVAS it is natural to ask
about decidability in general GVAS.

As mentioned in Section~\ref{sec:intro} there is a natural hierarchy of reachability-like problems for PVAS,
so also for GVAS. This is because the coverability problem for d-GVAS reduces to the reachability problem for d-GVAS
and the reachability problem for d-GVAS reduces to the reachability problem for $(d+1)$-GVAS.
Therefore, the next natural problem to consider would be the coverability problem for 2-GVAS. It is a challenging goal.

Another interesting question is about the complexity lower bounds for the reachability problem in GVAS.
Currently there are no known techniques, which allow for showing a complexity lower bound for d-GVAS higher
than immediately inherited from $(d+1)$-VASS (recall that the pushdown may simulate easily one additional VASS counter).
Therefore, in particular we only know \ackermann-hardness for the reachability problem for GVAS,
inherited from~\cite{DBLP:conf/focs/Leroux21,DBLP:conf/focs/CzerwinskiO21}. However, in~\cite{DBLP:conf/csl/LerouxPS14}
it is shown that the reachability set in GVAS can be finite, but of hyperAckermannian-size. That suggests that possibly
the reachability problem for GVAS is \hypackermann-hard.
We formulate here the conjecture, which is a kind of in the air among people working on the reachability-like problems
on VASS and its generalisations. Any step towards proving it or falsifying it would be very interesting.

\begin{conjecture}\label{conj:gvas-complexity}
The reachability problem for GVAS is \hypackermann-complete.
\end{conjecture}

\paragraph*{Semilinearity of the reachability sets and relations}
It is natural to ask whether the reachability sets and relations of 1-GVAS are semilinear.
Unfortunately, the reachability relations of 1-GVAS are not always semilinear.
It was shown in~\cite{DBLP:conf/csl/LerouxPS14} that already thin 1-GVAS do not have a semilinear reachability relation. A simple example is the following 1-GVAS having two nonterminals $X$ and $Y$ and the following four rules:
\[
Y \to 0 \hskip 1.5cm Y \to -1 \, Y \, 2 \hskip 1.5cm X \to 1 \hskip 1.5cm X \to -1 \, X \, Y.
\] 
It is easy to see that $a \trans{X} b$ if and only if $b \leq 2^a$,
which is clearly a non-semilinear relation.

However, it turns out that 1-GVAS with the top component being branching seem to be much closer
to having semilinear reachability relations.
We used semilinear sets a lot in our reasonings, and, in fact, we almost proved that top-branching 1-GVAS have semilinear reachability relations. In Section~\ref{sec:far-from-axis}, we showed how to compute a semilinear representation of their reachability relations far from the axis, that is, restricted to a set of the form $[B, \infty)^2$ for some $B \in \N$. Later, in Section~\ref{sec:mainproof}, we studied the reachability relation on lines close to the axes, or equivalently, reachability sets from small inputs. In one case (when the corresponding supertree has a successful supernode), we computed a semilinear representation of the reachability set from a certain threshold. Note that the threshold does not matter, as finite sets are semilinear sets, and we now have an algorithm for reachability in 1-GVAS. However, in the other case, we could only provide a thin under-approximation with same reachability set. This yields the following ``semilinearity transfer'' theorem.

\begin{theorem} \label{thm:semilin-transfer}
    If thin 1-GVAS have semilinear reachability sets, then top-branching 1-GVAS have semilinear reachability relations. Moreover, if the former are computable, then the latter also are.
\end{theorem}

Actually, this is even an equivalence because, for every $a \in \N$ and every thin 1-GVAS $H$ of initial nonterminal $S$, we can add a new initial nonterminal $S'$ and the two rules $S' \to -a S'S'$, $S' \to S$ on top of $H$ to construct a top-branching 1-GVAS with the same reachability set from the input $a-1$. Therefore (and by Theorem~\ref{thm:main}), thin 1-GVAS and top-branching 1-GVAS have the same reachability sets.

We formulate the following conjecture about reachability sets in thin 1-GVAS.

\begin{conjecture}\label{conj:thin-semilinear}
The reachability sets for thin 1-GVAS are semilinear and the semilinear representation is computable.
\end{conjecture}

By Theorem~\ref{thm:semilin-transfer}, Conjecture~\ref{conj:thin-semilinear} would imply semilinearity of the reachability relation in 1-GVAS with branching top component and computability of its reachability relation.
Such a result could improve a lot our understanding of the structure of reachability relation
for 1-GVAS and possibly allow for decidability of some other problems.
For example, it would show how to immediately decide whether two 1-GVAS have the same reachability relation.
Therefore, an interesting research direction is to prove or disprove Conjecture~\ref{conj:thin-semilinear}.

Notice that there is no hope for getting a small representation of the semilinearity set.
Adding to the above 1-GVAS a nonterminal $Z$ with two rules: $Z \to 0$ and $Z \to -1 Z Y$
we get that $a \trans{Z} b$ if and only if $b \leq \tower(a)$. Therefore the size of the reachability set, even with only three nonterminals,
can be non-elementary. Adding more nonterminals we can even get sets of sizes close to ackermannian.
Anyway, it is interesting to ask about semilinearity, even though there is no hope for small representations.

\paragraph*{Semilinearity in higher dimensions}
One might conjecture that the same property holds in some higher dimensions.
However, this is unfortunately not true, as shown by the following Example~\ref{ex:non-semilinear}.
Example~\ref{ex:non-semilinear} provides a 2-GVAS, in which the coverability relation is not semilinear.

\begin{example}\label{ex:non-semilinear}
Let $G$ be a 2-GVAS with two nonterminals $X$ and $Y$
and the following four rules:
\begin{align*}
X \to X \, X \, (-1,0) & \hskip 1cm X \to Y & \hskip 1cm Y \to (0,-2) \, Y  \, (1,1) & \hskip 1cm Y \to (0,0).
\end{align*}

\noindent
We present here only a sketch of an argument showing that the coverability relation for the above 2-GVAS
is not semilinear.
One can easily observe that $(0, a) \trans{Y} (b, c)$ if and only if $b + c = a$ and $b \leq a / 2$.
More generally $(d, a) \trans{Y} (b, c)$ if and only if $b + c = d + a$ and $b \leq d + a/2$.
Therefore $(0, a) \trans{X} (a/2, a/2)$ by the use of the second rule $X \to Y$.
However, to reach value higher then $a/2$ on the first counter one needs to use the first rule, which decreases
value on the first counter.
Notice that $(0, a) \trans{XX} (3a/4, a/4)$, $(0, a) \trans{XXX} (7a/8, a/8)$, etc., of course in the case when $a$
is divisible by $4$ or $8$. And more generally $(0, 2^a) \trans{X\cdots X} (2^a-1,1)$ if the number of $X$-es equals $a$.

We claim that in general, the highest $d$ for which $(0, 2^a) \trans{X} (d, 0)$ is $d = 2^a - a$.
The intuitive reason is that to transfer high value to the second counter we need many $X$-es,
but for production of each $X$ we need to decrease the first counter by $1$ (by the rule $X \to X \, X \, (-1,0)$).
Notice that indeed $(0, 2) \trans{X} (1, 1)$, but not $(0, 2) \trans{X} (2, 0)$.
Assuming that $(0, 2^a) \trans{X} (2^a - a, 0)$ the optimal derivation from $(0, 2^{a+1})$ would start from firing the rules
$X \to X \, X \, (0,-1) \to Y \, X \, (0, -1)$. Then
\[
(0, 2^{a+1}) \trans{Y} (2^a, 2^a) \trans{X} (2^a + (2^a-a), 0) \trans{(-1, 0)} (2^a + 2^a - a - 1,0) = (2^{a+1} - (a+1), 0).
\]
One can show that indeed $b = 2^{a+1} - (a+1)$ is the highest value $b$ such that $(b, 0)$ is reachable from $(0, 2^{a+1})$
and also therefore coverable.
Thus the set of $(x, y) \in \N^2$ such that $(0, x)$ covers $(y, 0)$ contains all the pairs $(2^a, 2^a-a)$ for $a \in \N$,
but does not contain any pair $(2^a, 2^a-a+1)$ for $a \in \N$. Such a set cannot be semilinear, which finishes the argument.
\end{example}

\subsubsection*{Acknowledgments}
We thank Filip Mazowiecki for helpful comments on the first draft of this paper.
We also thank Georg Zetzsche for numerous fruitful discussions about 1-GVAS reachability for recent years.

\bibliographystyle{alpha}
\bibliography{citat}

\appendix


\section{Proof of Claim~\ref{cl:finite-index-thin}}\label{sec:app1}
We recall the statement of Claim~\ref{cl:finite-index-thin}.

\vskip 0.3cm

\noindent
\textbf{Claim~\ref{cl:finite-index-thin}.}
Every thin context-free grammar is finite-index.

\vskip 0.3cm

\begin{proof}
Fix a finite-index grammar $G$ with the initial nonterminal $S$.
We prove by induction on the number $k$ of components in a thin grammar that its index is at most $k$.
For $k = 1$ the statement is obvious. Assume the induction assumption for $k$ and show it for $k+1$.
Take any $S$-derivation $\tau$. First applied rule is of the form $S \to X_1 S_1$ or $S \to S_1 X_1$ (recall that our rules are at most binary),
where $X_1$ is from component lower than $S$. By induction assumption we can create sequence of words starting from $X_1$,
which contains at most $k$ nonterminals each. Thus adding $S_1$ to each of those words makes a sequence
of words with at most $k+1$ nonterminals each. Finally, we get a word with only terminals and the nonterminal $S_1$,
we can continue expanding $S_1$. We continue similar way with the rule $S_1 \to X_2 S_2$ or $S_1 \to S_2 X_2$
and proceed the same further with nonterminals $S_2$, $S_3$, etc., all in the same component as $S$.
At some point the nonterminal $S_i$ will also be from component lower than $S$, which will finish creation of the sequence of
words and finalise the argument.
\end{proof}

\section{Proof of Lemma~\ref{lem:1GVAS-without-negative-cycle}}\label{sec:app2}
We recall the statement of Lemma~\ref{lem:1GVAS-without-negative-cycle}.

\vskip 0.3cm

\noindent
\textbf{Lemma~\ref{lem:1GVAS-without-negative-cycle}.}
Let $G$ be a 1-GVAS without negative cycles. Suppose that $G$ can cover the value $B \in \N$ from an input value $a < B$. Then, there is $b \in \N$ of size at most polynomial in $B$ and exponential in $\size(G)$ such that $a \xrightarrow{G} b$ and $B \le b$.

\vskip 0.3cm

\begin{proof}
    First, observe that the result is immediate if $G$ doesn't have positive cycles either.
    Indeed, all the possible effects of $G$ are generated by simple derivations, and there are finitely many simple derivations. Let $\Delta$ be the maximum of those effects. If, from input $a$, $G$ can cover $B$, then it can reach some $b \in \N$ such that $B \le b \le a + \Delta$. From $a < B$, we deduce that $b \le B + \Delta$, which is obviously a bound polynomial in $B$ and exponential in $\size(G)$.

    We do the rest of the proof by induction on the depth $d$ of the dag of components of nonterminals.
    Recall from Section~\ref{sec:triangle} that a derivation $\theta$ is irreducible if removing any simple cycle from $\theta$ decreases the set of nonterminals it contains. 
    For the rest of the proof, we let $m(G)$ denote the smallest natural number such that 
    \begin{itemize}
    \item all the irreducible derivations have effect at most $m(G)$ in absolute value and are all valid from input $m(G)$
    \item simple cycles of positive effect have effect at most $m(G)$.
    \end{itemize}
    Value $m(G)$ is at most exponential in $\size(G)$ (Claim~\ref{cl:irreducible-size} guarantees that irreducible derivations are of size at most exponential in $\size(G)$). Since $G$ doesn't contain any negative cycle, minimal effects are reached at simple derivations. Therefore, $m(G)$ also has the property that all derivations produced by $G$, rooted at any nonterminal (non necessarily the initial one), have effect at least $-m(G)$.

    To prove Lemma~\ref{lem:1GVAS-without-negative-cycle}, we show the following claim by induction on $d$.
    Notice that Claim~\ref{cl:bound-on-jump} immediately implies Lemma~\ref{lem:1GVAS-without-negative-cycle}.

    \begin{claim}\label{cl:bound-on-jump}
    Let $G$ be a 1-GVAS without positive cycle.
    If $G$ can cover $B \in \N$ from an input $a < B$, then there is $b \in \N$ such that $a \xrightarrow{G} b$ and $ B \le b \le b_{\text{max}}(B,d,G)$, where $b_{\text{max}}(B,d,G) :=
    m(G) \cdot (\lvert \mathcal N(G) \rvert + 1) \cdot (B + 2 m(G) \lvert \mathcal N (G) \rvert) + 2(d+1) \cdot m(G) \cdot (\lvert \mathcal N (G) \rvert + 1)$
    \end{claim}

    \begin{proof}

    We take as basic case the case where $G$ doesn't have any positive cycle, which we already treated. We view nonterminals as GVAS with $d=-1$ (and with no positive cycle). For the induction step, we take a 1-GVAS $G$ without positive cycle, but with negative cycles, and with depth of dag of components of nonterminals $d + 1$

    Let $\theta$ be a derivation with counters of input $a$ and of output at least $B$.
    We take $\theta$ with the smallest possible number of nodes labelled by nonterminals in the top components.
    Our goal is to modify $\theta$ such that its new output value is lower than $b_{\text{max}}(B,d+1,G)$ (while $\theta$ remains a valid derivation with counters produced by $G$.)
    We denote by $\pref(\theta)$ the largest prefix of $\theta$ such that all internal nodes are labelled by a top nonterminal and have no strict ancestor with the same label.
    In other words, $\pref(\theta)$ is obtained from $\theta$ by starting at the root, going down all the branches until a node is found that is labelled by a lower nonterminal or a nonterminal already seen in an ancestor, and removing the strict descendants of such nodes. Observe that $\pref(\theta)$ has depth at most $\lvert \text{Top}(G) \rvert$.
    In the case where the top component is thin, the internal nodes of $\pref(\theta)$ simply form a path. Thus, the thin case is simpler than the branching one, but we treat them simultaneously.
    Our proof strategy consists of trying to bound the counter values of the leaves of $\pref(\theta)$,
    starting from the left and moving to the right. 

    For that, we fix a threshold $T(B,G) := B + \lvert \text{Top}(G) \rvert \cdot m(G)$ and with it we define the \emph{no-jump property}.
    Notice that $B \leq T(B, G) \leq b_{\text{max}}(B,d+1,G)$.
    The no-jump property at node $N$ states that, if the input of $N$ is (strictly) less than $T(B,G)$, then its output also is. 
    The input value of $\theta$ is $a < B \le T(B,G)$. Consequently, if in $\pref(\theta)$ all the leaves strictly on the left of some leaf $N$ have the no-jump property, then the input of $N$ is also smaller than $T(B,G)$ (as the input of a leaf of $\pref(\theta)$ is the output of the leaf on its left).

    \begin{claim} \label{cl:no-jump-lower}
        If a leaf of $\pref(\theta)$ labelled by a lower nonterminal $V$ doesn't have the no-jump property
        then we can modify derivation $\theta$ to $\theta'$ with an output in between $B$ and $b_{\text{max}}(B,d,G)$.
    \end{claim}

    \begin{proof}[Proof of Claim \ref{cl:no-jump-lower}]


    Suppose that $N$ is a leaf of $\pref(\theta)$ with input $a' < T(B,G)$ and output at least $T(B,G)$. By induction hypothesis there is valid derivation with counters, produced by $G_V$, of input $a'$ and output $b'$, with $b' \le b_{\text{max}}(B,d,G_V)$.
We can replace the subtree rooted at $N$ by this one; and replace by simple derivations the subtrees rooted at nodes that are right children of strict ancestor than $N$, but not on the path from the root to $N$. This gives a valid derivation with counters, produced by $G$, of input $a$, and with a output at most $b' +  \lvert \text{Top}(G) \rvert \cdot m(G)$, which is clearly smaller than $b_{\text{max}}(B,d+1,G)$
    \end{proof}


    Recall that our goal is to modify derivation with counters $\theta$ such that its output remains above $B$, but falls below $b_{\text{max}}(B,d+1,G)$. For that, we consider a prefix $\pref(\theta)$ of $\theta$, whose leaves are the first lower nonterminals, or repeated top nonterminals, encountered when moving down from the root. We want to bound the counters of the leaves of $\pref(\theta)$. To this end, we defined the no-jump property, that we have just proved for leaves of $\pref(\theta)$ labelled by lower nonterminals.

    Now, we focus on the leaves of $\pref(\theta)$ labelled by top nonterminals. We also try to prove that they satisfy the no-jump property. In some cases, we can't, but then we directly provide the modified version $\theta'$ of $\theta$, with output comprised between $B$ and $b_{\text{max}}(B,d+1,G)$. On the other hand, if we all the leaves of $\pref(\theta)$ satisfy the no-jump property, then the output of $\theta$ is at most $T(B,G)$, which concludes the proof.

We call $N$ the leaf of $\pref(\theta)$ on which we are working, and $N'$ its ancestor with same label.
Notice that $N'$ is unique since $N$ is a leaf of the prefix of the derivation such no nonterminal repeats on any path above leaves.
We start from the left, so that we can assume that all the leaves of $\pref(\theta)$ strictly on the left of $N$ satisfy the no-jump property. We distinguish three cases, depending of the left, right and global effects of the cycle $\gamma$ between $N'$ and $N$.
    \begin{itemize}
        \item If the cycle $\gamma$ has nonpositive left effect, then the minimality property of $\theta$ implies that the output of $N$ is smaller than $T(B,G)$ (which is of course stronger than the no-jump property). Indeed, if the output of $N$ is at least $T(B,G)$, and if we remove the cycle between $N'$ and $N$, then we obtain a smaller derivation with counters from $a$ to an output above $B$. (Recall that \emph{all derivations} produced by $G$, with any nonterminal at the root, have the effect greater than $m(G)$. This is in particular the case for the subtrees rooted at right children of strict ancestors of $N'$, and there at most $\lvert \text{Top}(G) \rvert$ of them to cross before reaching the output of $\theta$.)
        \item If the cycle $\gamma$ has global effect zero, but positive left effect then we can add many copies of $\gamma$
        without any global effect, but increasing the input and output of $N$ arbitrarily. The intuitiion
        is that then at node $N$ we do not need to care whether the counter drops below zero, so it is easier
        to produce derivations with not too big output. Recall also that we may assume existence of a cycle $\gamma_+$
        with small, but strictly positive effect. Roughly speaking we construct $\theta'$ by pumping the cycle
        $\gamma$ many times, making all the other subtrees as small as possible and pasting the cycle $\gamma_+$
        an appropriate number of times. Precisely speaking we construct $\theta'$ by doing the following modifications to $\theta$:
        \begin{itemize}
            \item Replace by simple derivations (produced by $G$) all the subtrees rooted at right children of strict ancestors of $N'$ that are not ancestors of $N'$.
            \item Replace the subtree rooted at $N$ by an irreducible derivation (produced by $G$) which contains a nonterminal $V$ such that there is a simple $V$-cycle $\gamma_+$ of positive global effect.
            \item Insert, in the new subtree rooted at $N$, the smallest number of copies of $\gamma_+$ such that the counter don't fall below zero on the right of the subtree rooted at $N'$, and such that the output value of the whole derivation is at least $B$.
            \item Insert at $N'$ enough copies of the cycle $\gamma$ to ensure that the counter don't fall below zero in the subtree rooted at $N'$. (This doesn't change the output value of the derivation, because $\gamma$ has global effect zero.)
        \end{itemize}
        \item If the cycle $\gamma$ has positive left and global effects, then we first replace in $\theta$ all the subtrees rooted at a node on the right of the branch from $N'$ to $N$ (the main branch of $\gamma$) by simple derivations (produced by $G$). Our strategy is the following. We know that the left effect of $\gamma$ is small. Indeed, since the input and output counters of every leaf of $\pref(\theta)$ strictly on the left of $N$ is at most $T(B,G)$, we know that the left effect of $\gamma$ is the difference between two nonnegative values smaller as $T(B,G)$, and thus smaller than $T(B,G)$. After our modification of $\gamma$, the right effect becomes also small. The global effect of $\gamma$ remains nonnegative (since, by assumption, $G$ doesn't have negative cycle). If it becomes zero after the modification, we can do as in the previous case. Now we assume that it is positive. To construct $\theta'$, we replace by simple derivations (produced by $G$) the subtree rooted at $N$ and the subtrees on the right of the branch from the root to $N$. Then, we insert at $N$ the least number of copies of $\gamma$ necessary to make the whole derivation valid and to have an output $b \ge B$.
        
        To upperbound the output $b$, we first observe that $(\lvert \text{Top}(G) \rvert + 1)m(G)$ copies are enough to make $\theta'$ valid (including the one already present in $\theta$). Indeed, remember that both the left and the global effect of $\gamma$ are at least one. Thus, if we insert $(\lvert \text{Top}(G) \rvert + 1)m(G)$ copies of $\gamma$ above $N'$, then the input of $N'$ is at least $(\lvert \text{Top}(G) \rvert + 1)m(G)$, and its output at least $\lvert \text{Top}(G)\rvert m(G)$. Then, we go through the right part of all the copies of $\gamma$. The output counter after each copies remains at least $\lvert \text{Top}(G)\rvert m(G)$. (That is because both the left and global effects of $\gamma$ are at least one. Going through the right part of a copy of $\gamma$ means removing to the counter the left effect of $\gamma$ and adding its global effect.) Finally, $N$ is at depth at most $\lvert \text{Top}(G) \rvert$, so there are at most this number of simple derivations to go through before reaching the output of $\theta'$, all of which have effect at least $-m(G)$ and are valid from input $m(G)$.

        If $(\lvert \text{Top}(G) \rvert + 1)m(G)$ copies are also enough for $b$ to exceed $B$, then we insert this number of copies and we have $b \le T(B,G) + (\lvert \text{Top}(G) \rvert + 1)m(G) \cdot (T(B,G) + \lvert \text{Top}(G) \rvert m(G)) + \text{Top}(G) \rvert m(G)$, which is a small enough upperbound to finish the proof.
        Otherwise, we keep inserting more copies of $\gamma$ until the output exceeds $B$.
        We cannot exceed $B$ by more than the effect of $\gamma$, therefore $b \le B + T(B,G)$, which is also sufficient to finish the proof.
    \end{itemize}

\end{proof}

\end{proof}

\end{document}